\documentclass[11pt,a4paper]{amsart}

\usepackage[utf8]{inputenc}
\usepackage[T1]{fontenc}
\usepackage[english]{babel}
\usepackage{amsmath,amssymb,amsthm}
\usepackage{mathtools}
\usepackage{hyperref}
\hypersetup{colorlinks=true,linkcolor=blue,citecolor=blue,urlcolor=blue}
\usepackage{geometry}
\usepackage{enumitem}
\usepackage{mathrsfs}

\theoremstyle{plain}
\newtheorem{theorem}{Theorem}[section]
\newtheorem{proposition}[theorem]{Proposition}
\newtheorem{lemma}[theorem]{Lemma}
\newtheorem{corollary}[theorem]{Corollary}
\newtheorem{conjecture}[theorem]{Conjecture}

\newtheorem{mainthm}{Theorem}
\newcommand{\mainthmletter}[1]{\renewcommand{\themainthm}{#1}}

\theoremstyle{definition}
\newtheorem{definition}[theorem]{Definition}
\newtheorem{example}[theorem]{Example}
\newtheorem{remark}[theorem]{Remark}
\newtheorem{assumption}[theorem]{Assumption}

\newcommand{\M}{\mathcal{M}}
\newcommand{\Mstar}{\mathcal{M}_*}
\newcommand{\Hcal}{\mathcal{H}}
\newcommand{\Lcal}{\mathcal{L}}
\newcommand{\Acal}{\mathcal{A}}
\newcommand{\Dcal}{\mathcal{D}}
\newcommand{\tr}{\tau}
\newcommand{\rhohat}{\widehat{\rho}}
\newcommand{\Jcal}{\mathcal{J}}
\newcommand{\dd}{\mathrm{d}}
\DeclareMathOperator{\Ent}{Ent}
\DeclareMathOperator{\grad}{grad}
\DeclareMathOperator{\dive}{div}
\DeclareMathOperator{\dom}{dom}

\title[Arnold--Nielsen Geometry for Complexity-Deformed Noncommutative Transport]
{Arnold--Nielsen Geometry for Complexity-Deformed Noncommutative Transport}

\author{Alberto Acevedo \and Antonio Falcó}
\address{Department of Mathematics, Physics and Technological Sciences, Universidad CEU Cardenal Herrera, Spain}
\email{alberto.acevedomelendez@uchceu.es}
\email{afalco@uchceu.es}

\date{\today}

\begin{document}

\begin{abstract}
We deform the Carlen--Maas--Wirth framework for noncommutative dynamical optimal transport by an Arnold--Nielsen type complexity operator. A positive state-independent operator $G$ compatible with the Hilbert bimodule structure of a noncommutative differential calculus $\partial\colon\M\to\Hcal$ can be absorbed into the calculus itself,
\[
\partial_G:=G^{1/2}\partial.
\]
The corresponding complexity-weighted transport problem is exactly the unweighted transport problem generated by $\partial_G$, whenever the deformed quadratic form remains Dirichlet. In finite dimensions we prove existence of minimizers for density-dependent Petz-class metrics and for fixed physical complexity weights, the latter without commutation between $G$ and the state-dependent mobility. On unitary orbits we identify the induced distance with a quotient metric coming from a right-invariant complexity geometry. This yields an exact Bell-state preparation result via Clairaut's relation and an exactly computed restricted-path upper bound for GHZ preparation; the Lindblad detailed-balance case is included only as entropy-gradient-flow background.
\end{abstract}


\maketitle


\section{Introduction}
\label{sec:intro}

This paper develops an Arnold--Nielsen geometric deformation of noncommutative dynamical optimal transport. The starting point is that a complexity tensor should not merely penalize a fixed action: under a natural bimodule-compatibility condition it can be absorbed into the differential calculus that generates transport. This gives new Wasserstein-type geometries on von Neumann algebraic state spaces and, on unitary orbits, quotient metrics induced by right-invariant complexity geometries on compact Lie groups. The guiding principle is therefore geometric rather than only variational: complexity changes the differential structure, and the corresponding orbit dynamics becomes a quotient form of Arnold--Nielsen geodesic motion.

The classical model behind this viewpoint is the Benamou--Brenier formalism~\cite{BenamouBrenier2000}, which rewrites the squared Wasserstein-2 distance between two probability measures $\rho_0,\rho_1$ on $\mathbb{R}^n$ as the infimum of a kinetic action,
\begin{equation}
\label{eq:BB}
W_2(\rho_0,\rho_1)^2 = \inf_{\rho,v}\left\{ \int_0^1 \int_{\mathbb{R}^n} |v_t(x)|^2\,\rho_t(x)\,\dd x\,\dd t \; : \; \partial_t \rho_t + \dive(\rho_t v_t) = 0,\; \rho_{t=0}=\rho_0,\; \rho_{t=1}=\rho_1 \right\}.
\end{equation}
This dynamical formulation has motivated, over the last two decades, an active search for noncommutative analogues in which the pair $(\rho, v)$ is replaced by objects defined over a $C^*$-algebra or a von Neumann algebra $\M$. The work of Carlen and Maas~\cite{CarlenMaas2014,CarlenMaas2017,CarlenMaas2020,CarlenMaas2020Corr} and of Wirth~\cite{Wirth2018}, building on the theory of Dirichlet forms and derivations due to Cipriani and Sauvageot~\cite{CiprianiSauvageot2003}, shows that certain symmetric quantum Markov semigroups — in particular, those satisfying a detailed-balance condition — can be interpreted as gradient flows of relative entropy with respect to a transport metric built from a noncommutative differential calculus. In this framework the calculus
\[
\partial:\mathcal M\longrightarrow\mathcal H
\]
plays the role of a noncommutative gradient, the adjoint $\partial^*$ plays the role of a divergence, and the logarithmic mean $\widehat\rho$ supplies the noncommutative replacement for multiplying a velocity field by a density. The resulting functional inequalities and mixing-time estimates have since been used to control the entropy production of dissipative dynamics~\cite{DattaRouze2020} and, more recently, to certify quasi-optimal preparation of quantum Gibbs states~\cite{CapelEtAl2025}.

This dynamical, Dirichlet-form route is one of several inequivalent ways of quantizing the Benamou--Brenier and Kantorovich formulations of optimal transport. Static, coupling-based quantizations of the Wasserstein distance between density matrices have been developed by Caglioti, Golse and Paul~\cite{CagliotiGolsePaul2023}, by Cole, Eckstein, Friedland and \.Zyczkowski~\cite{ColeEtAl2021} via semidefinite programming over bipartite couplings, and, using a channel-based transport plan, by De Palma and Trevisan~\cite{DePalmaTrevisan2021}; a convex-regularization treatment of the resulting static problem is given in~\cite{CaputoEtAl2024}. None of these static formulations comes with a continuity equation of the type~\eqref{eq:continuity} below, which is specific to the dynamical route pursued here. The metric geometry of density matrices itself is classified, at the level of Riemannian (rather than transport) metrics, by Petz's theorem on monotone metrics~\cite{PetzSudar1996}, which underlies the Kubo--Ando family of density-dependent complexity metrics of Section~\ref{ssec:petz-existence}. Beyond density matrices, dynamical and static quantum transport distances have also been constructed on discrete quantum groups; in particular, Anshu, Jekel and Landry~\cite{AnshuJekelLandry2025} recently introduced Wasserstein-type distances on the quantum permutation group $S_n^+$, generalizing the Hamming metric on $S_n$ rather than the Riemannian geometry of $\M$-states pursued here — a structurally different but closely related instance of the same broad programme of endowing noncommutative probabilistic data with a transport geometry.

The central thesis of this paper is a single deformation principle:
\emph{a compatible complexity operator deforms the noncommutative differential calculus itself}. More precisely, for a state-independent positive operator $G$ compatible with the left and right module actions and the real structure of $\mathcal H$, the operator
\[
\partial_G:=G^{1/2}\partial
\]
is again a closed real derivation. When the associated quadratic form remains Markovian, it defines a new conservative Dirichlet form and hence a new dynamical transport geometry. In this sense complexity is internalized into the geometry, instead of being imposed only as a cost on an unchanged continuity equation.

This principle also has an Arnold--Nielsen interpretation. On unitary orbits, the same type of complexity tensor induces a right-invariant metric on the relevant compact Lie group, and the transport distance between two orbit states is the quotient distance from the identity to the stabilizer coset of the target. Thus the construction connects the Carlen--Maas--Wirth transport calculus with the geometric-mechanics idea that dynamics is geodesic motion for an invariant metric, and with Nielsen's view of quantum complexity as anisotropic geometry on unitary groups. The qubit and entangling examples, the Petz-class existence theory, and the dissipative sector are organized around making this thesis precise and delimiting exactly where it does and does not yet extend.

Concretely, the general formulation of the \emph{noncommutative continuity equation} developed below on a tracial von Neumann algebra $(\M,\tau)$ is complemented by a logically distinct unitary-orbit control problem governed by the Liouville--von Neumann equation $\dot\rho_t = -i[H_t,\rho_t]$ (Section~\ref{sec:hamiltonian}); it naturally incorporates a penalization operator $G_\rho$ acting as a \emph{complexity metric} on the directions of the bimodule $\Hcal$, penalizing certain derivations or degrees of locality (Section~\ref{sec:complexity}); it admits a distinguished class of state-independent, bimodule-compatible complexity operators for which the penalty can be absorbed into the differential calculus itself, producing a new transport geometry on a general finite von Neumann algebra (Section~\ref{sec:calculus-deformation}); and it admits the dissipative Lindblad dynamics as a special, and separately motivated, case, reinterpreted as a gradient flow of relative entropy with respect to a transport metric (Section~\ref{sec:lindblad}).

\subsection*{Organization}

Section~\ref{sec:transport} fixes the mixed-state transport formalism: states, the noncommutative continuity equation, the operator mean $\rhohat$, and the dynamical distance $W_\partial$. Section~\ref{sec:complexity-actions} introduces complexity metrics and the calculus deformation $\partial_G$, then compares the cost-only and calculus-deformed variational problems and checks the classical Benamou--Brenier limit. Section~\ref{sec:existence} proves finite-dimensional existence of minimizers for the unpenalized metric, for Petz-class density-dependent metrics, and for fixed physical weights. Sections~\ref{sec:orbits}--\ref{sec:bell-example} develop the unitary-orbit geometry, including the quotient metric theorem, the anisotropic qubit, Bell-state preparation, and the restricted GHZ upper bound. Section~\ref{sec:discussion-main} records the detailed-balance Lindblad background and the remaining open questions.

\subsection*{Prior work reused}

Sections~\ref{sec:duality}--\ref{sec:continuity} and~\ref{sec:rhohat}--\ref{sec:action} fix notation and record, in a form adapted to the present setting, constructions due to Cipriani and Sauvageot~\cite{CiprianiSauvageot2003} (the differential calculus $(\Hcal,\partial)$) and to Carlen and Maas~\cite{CarlenMaas2014,CarlenMaas2017,CarlenMaas2020,CarlenMaas2020Corr} and Wirth~\cite{Wirth2018} (the operator logarithmic mean $\rhohat$, the dynamical distance $W_\partial$, and their gradient-flow interpretation of Lindblad dynamics). Theorem~\ref{thm:gradient-flow} in Section~\ref{sec:lindblad} is an informal restatement of their results, included for completeness and not claimed as new; likewise, the joint-convexity property underlying Step~3 of Theorem~\ref{thm:existence} is cited, not reproved, from~\cite{CarlenMaas2017,CarlenMaas2020}.

Compared with the authors' dilation-based complexity framework~\cite{AcevedoFalco2026}, the present paper has a different mathematical object and new variational content. The earlier work studies geometric costs of unitary and weak-coupling dilations of quantum evolutions. Here the new contributions are the absorption of compatible complexity weights into a noncommutative differential calculus (Theorems~\ref{thm:deformed-derivation}--\ref{thm:complexity-deformation}), the finite-dimensional existence theorem for fixed noncommuting physical weights (Theorem~\ref{thm:existence-fixed-G}), and the exact Bell-state quotient-geodesic computation (Theorem~\ref{thm:bell-global}). We use~\cite{AcevedoFalco2026} only to motivate the choice of complexity weights in the Hamiltonian and Lindblad sectors and to import the dilation estimate behind Proposition~\ref{prop:lindblad-bound}; those ingredients are not presented here as new transport theorems.

\subsection*{Main results and contributions}

The main results are organized into four blocks, each stated and proved in full in the section indicated; the existence block contains several complementary finite-dimensional theorems. Everything else is either a worked example, an existence-theoretic refinement of one of these blocks, or a proposal explicitly flagged as such, and the hypotheses each result depends on are stated with it rather than left implicit.

Theorem~\ref{thm:deformed-derivation}, proved in Section~\ref{sec:calculus-deformation}, has an unconditional part and a conditional part. Under bimodule compatibility, reality and uniform positivity of $G$ (Definition~\ref{def:compatible-complexity}), the deformed operator $\partial_G=G^{1/2}\partial$ is a closed real derivation with the same domain as $\partial$. It defines a new conservative Dirichlet form only under the further, non-automatic hypothesis that the resulting quadratic form is Markovian (Assumption~\ref{ass:markov-compatible-G}, for which Proposition~\ref{prop:diagonal-deformation} gives a verifiable sufficient condition). Theorem~\ref{thm:complexity-deformation}, in the same section, shows that the complexity-weighted transport problem for such a $G$ is exactly the unweighted transport problem generated by $\partial_G$ — an equivalence about the deformed continuity equation $\dot\rho_t+\partial^*(G\rhohat_tv_t)=0$ and action~\eqref{eq:action-identification}, and not about the cost-only penalization of Section~\ref{sec:complexity}, which keeps the original continuity equation~\eqref{eq:continuity} fixed (Remark~\ref{rem:cost-versus-calculus-deformation}). Together with the bi-Lipschitz comparison for fixed cost weights on the unchanged continuity equation (Proposition~\ref{prop:fixed-G-comparison}) and, under Wirth's regularity and gradient-estimate hypotheses~\cite{Wirth2018}, a conditional geodesic-existence corollary (Corollary~\ref{cor:deformed-geodesics}), these two theorems constitute the central contribution of this paper: a compatible complexity operator deforms the calculus itself, rather than merely reweighting a fixed transport problem.

The existence results of Section~\ref{sec:existence} provide the variational support for the mixed-state part of the theory. Theorem~\ref{thm:existence} proves finite-dimensional existence for the unweighted metric $W_\partial$ by the direct method in momentum variables. Theorem~\ref{thm:existence-petz} extends this to the Petz class of density-dependent complexity metrics built from range-regular Kubo--Ando means. Theorem~\ref{thm:existence-fixed-G} then treats fixed physical weights: even when $G$ is diagonal in a basis unrelated to the eigenbasis of $\rho_t$, and the momentum cost is not jointly convex, minimizers still exist by lower semicontinuity and convexity in the momentum variable. Kubo--Ando means without range-regularity, such as the arithmetic mean, instead define an alternative mobility rather than a complexity weight on the original one (Example~\ref{ex:arithmetic-mean}).

Theorem~\ref{thm:orbit-geometry}, proved in Section~\ref{sec:bridge}, identifies the distance between two states on the same unitary orbit with the distance, in a right-invariant complexity geometry on the ambient compact Lie group, from the identity to the coset of the stabilizer of the initial state, and shows that this distance is attained, by compactness and the Hopf--Rinow theorem. This unitary-orbit sector is motivated by, but not derived from, the three theorems above: it is governed directly by the Liouville--von Neumann equation (Remark~\ref{rem:spectrum-preserved}) and does not invoke the operator mean $\rhohat$. We illustrate Theorem~\ref{thm:orbit-geometry} with two exact examples: a genuine global-optimality theorem, via Clairaut's relation, for an entangling $\mathfrak{su}(2)$ subsystem connecting a product state to a Bell state (Theorem~\ref{thm:bell-global}), and its extension to an exactly computed, exponentially weighted upper bound — not a lower bound, and not a proof of exponential complexity — for a restricted $N$-qubit GHZ preparation path (Theorem~\ref{thm:ghz-generic}, Section~\ref{ssec:ghz-scope}).

Beyond these main result blocks, we propose, without fully resolving, dilation-motivated complexity weights $G_\rho$ for the Hamiltonian sector (Section~\ref{sec:bridge}) and, more tentatively, for the Lindblad sector (Section~\ref{ssec:lindblad-complexity}), obtained from the geometric complexity of unitary dilations of~\cite{AcevedoFalco2026}; the Lindblad proposal is stated as a definition and a dilation-based cost bound (Proposition~\ref{prop:lindblad-bound}), not as an intrinsic existence theory (Remark~\ref{rem:lindblad-scope}). Section~\ref{sec:discussion} delimits precisely what remains open for the calculus-deformation, existence, unitary-orbit and Lindblad components.

\section{Noncommutative dynamical transport}
\label{sec:transport}

This section fixes the mixed-state transport formalism used throughout the paper: the predual duality between observables and states, the noncommutative differential calculus and its associated continuity equation, the operator mean $\rhohat$, and the resulting action functional and dynamical distance $W_\partial$ in both velocity and momentum variables. None of the results in this section is new; see the ``Prior work reused'' paragraph of Section~\ref{sec:intro}.

\subsection{States and dual dynamics}
\label{sec:duality}

Let $\M$ be a von Neumann algebra and let $\Mstar$ be its predual. We consider a family of normal states
\[
\omega_t \in \Mstar^+, \qquad \omega_t(\mathbf{1}) = 1.
\]

\begin{assumption}
\label{ass:trace}
There exists a normal, faithful, semifinite trace $\tau$ on $\M$, so that every normal state admits a density in $L^1(\M,\tau)$: there is $\rho_t \in L^1(\M,\tau)_+$ with $\tau(\rho_t)=1$ such that
\[
\omega_t(a) = \tau(\rho_t\, a), \qquad a \in \M.
\]
\end{assumption}

Under Assumption~\ref{ass:trace}, the dynamics can be described equivalently at the level of states $\omega_t$ or of densities $\rho_t$.

An evolution of observables is given by a family (strongly continuous, in the appropriate sense) of linear maps
\[
\Lcal_t \colon \M \longrightarrow \M.
\]
The Heisenberg equation for an observable $a_t$ is
\begin{equation}
\label{eq:heisenberg}
\frac{\dd}{\dd t} a_t = \Lcal_t(a_t),
\end{equation}
and the dual equation, on the predual $\Mstar$, is
\begin{equation}
\label{eq:dual-state}
\dot\omega_t = \Lcal_{t*}\,\omega_t,
\end{equation}
where $\Lcal_{t*}\colon \Mstar \to \Mstar$ is the predual of $\Lcal_t$, characterized by
\[
(\Lcal_{t*}\omega)(a) = \omega(\Lcal_t a), \qquad a \in \M,\; \omega \in \Mstar.
\]
In terms of densities, \eqref{eq:dual-state} reads
\begin{equation}
\label{eq:dual-density}
\dot\rho_t = \Lcal_{t*}(\rho_t).
\end{equation}

\subsection{Noncommutative continuity equation}
\label{sec:continuity}

\subsubsection{Noncommutative differential calculus}

Let $\Hcal$ be an $\M$-Hilbert bimodule and let
\[
\partial \colon \M \longrightarrow \Hcal
\]
be a (closable, densely defined) derivation, that is,
\begin{equation}
\label{eq:leibniz}
\partial(ab) = (\partial a)\,b + a\,(\partial b), \qquad a,b \in \dom(\partial).
\end{equation}
We denote by $\partial^*\colon \Hcal \to \M$ its adjoint with respect to the natural inner products induced by $\tau$ on $\M \subset L^2(\M,\tau)$ and on $\Hcal$; $\partial^*$ plays the role of \emph{noncommutative divergence}.

\begin{definition}[Continuity equation]
\label{def:continuity}
Given a velocity field $v_t \in \Hcal$ and a density $\rho_t \in L^1(\M,\tau)_+$, we say that the pair $(\rho_t,v_t)_{t\in[0,T]}$ satisfies the \emph{noncommutative continuity equation} if
\begin{equation}
\label{eq:continuity}
\dot\rho_t + \partial^*\bigl(\rhohat_t\, v_t\bigr) = 0,
\end{equation}
where $\rhohat_t$ is the noncommutative multiplication operator associated with $\rho_t$ (Section~\ref{sec:rhohat}); implicit in~\eqref{eq:continuity} is the requirement $\rhohat_tv_t\in\dom(\partial^*)$ for a.e.\ $t$, which we take as part of the definition of an \emph{admissible curve} rather than as an automatic consequence of $v_t\in\Hcal$. In the finite-dimensional setting of Sections~\ref{sec:action} onward, $\partial^*$ is bounded and everywhere defined, so this requirement is automatic and imposes no restriction; in the general semifinite setting it is a genuine constraint, which the weak form of the next subsection (eq.~\eqref{eq:weak-form}) is designed to bypass, since it only requires $v_t\in\Hcal$ and $a\in\dom(\partial)$, not $\rhohat_tv_t\in\dom(\partial^*)$.
\end{definition}

Equation~\eqref{eq:continuity} is the noncommutative analogue of the classical continuity equation from transport theory. In the commutative specialization $\M = L^\infty(\mathbb{R}^n)$, $\Hcal = L^2(\mathbb{R}^n;\mathbb{R}^n)$, $\partial = \nabla$, integration by parts gives $\partial^*=-\dive$, so~\eqref{eq:continuity} becomes $\dot\rho_t-\dive(\rho_tv_t)=0$, i.e.\ the standard continuity equation for the relabeled velocity field $-v_t$. This sign convention is fixed throughout.

\subsubsection{Weak form}

The most robust formulation — since it does not require every state to admit a bounded density in $\M$ — is the weak form: for every self-adjoint $a=a^*\in\dom(\partial)$, using the adjoint convention $\langle\partial a,\xi\rangle_{\Hcal} = \langle a,\partial^*\xi\rangle_{L^2(\M,\tau)}$, a direct computation from Definition~\ref{def:continuity} and the self-adjointness of $\rhohat_t$ gives
\begin{equation}
\label{eq:weak-form}
\frac{\dd}{\dd t}\,\omega_t(a) = -\bigl\langle \partial a,\, \rhohat_t\, v_t \bigr\rangle_{\Hcal},
\end{equation}
Writing $\langle\partial a,\rhohat_tv_t\rangle$ rather than $\langle v_t,\rhohat_t\partial a\rangle$ avoids any need to invoke conjugate-symmetry of the (sesquilinear, in general) inner product on $\Hcal$ to pass between the two, an issue that only disappears automatically when both arguments are real. Identity~\eqref{eq:weak-form} is the noncommutative analogue of
\[
\frac{\dd}{\dd t}\int a\,\rho_t\,\dd x = -\int \nabla a \cdot v_t\, \rho_t\,\dd x,
\]
i.e.\ the same relabeled classical continuity equation described above.

\subsection{The noncommutative multiplication \texorpdfstring{$\rhohat$}{rho-hat}}
\label{sec:rhohat}

In the commutative case, ``multiplying the velocity field by the density'' is simply $\rho v$. In the noncommutative case this operation is not uniquely determined, since in general $\rho v \neq v\rho$. We adopt the \emph{operator logarithmic mean}:
\begin{equation}
\label{eq:rhohat-def}
\rhohat(v) = \int_0^1 \rho^{s}\, v\, \rho^{1-s}\, \dd s = \int_0^1 L_{\rho^s}R_{\rho^{1-s}}(v)\, \dd s,
\end{equation}
where $L_a,R_a$ denote left and right multiplication by $a\in\M$ on $\Hcal$. This presupposes that $\Hcal$ carries commuting left and right actions of $\M$ under which $L_a,R_a$ are well defined for the powers of $\rho$ appearing above. If $\rho\in L^1(\M,\tau)_+$ is not bounded, then $\rho^s$ need not lie in $\M$, and $L_{\rho^s},R_{\rho^{1-s}}$ need not be bounded operators on $\Hcal$. Thus~\eqref{eq:rhohat-def} must, in general, be understood as a closed quadratic form, finite exactly where the integral converges, rather than as an everywhere-defined bounded operator.

We do not pursue this domain analysis in general. Every result of this paper that uses $\rhohat$ substantively (Sections~\ref{sec:action}--\ref{sec:complexity} onward) is proved for $\M=M_n(\mathbb C)$, where $\rho$ is automatically bounded and~\eqref{eq:rhohat-def} is an everywhere-defined bounded operator (Example~\ref{ex:finite-dim}). The structural results of Section~\ref{sec:calculus-deformation} are the exception: they concern the calculus itself and hold on a general finite von Neumann algebra. Wherever $\rhohat$ appears explicitly there (Definition~\ref{def:G-continuity}), the same domain caveat applies. The definition above is stated for a general semifinite $(\M,\tau)$ only to fix notation and to make transparent which finite-dimensional identities are, in principle, extendable.

\begin{example}[Finite dimension]
\label{ex:finite-dim}
If $\rho = \sum_k \lambda_k \lvert k\rangle\langle k\rvert$ is the spectral decomposition of $\rho$, then, in the basis $\{|k\rangle\}$,
\begin{equation}
\label{eq:rhohat-matrix}
\bigl(\rhohat(v)\bigr)_{jk} = \Lambda(\lambda_j,\lambda_k)\, v_{jk}, \qquad \Lambda(x,y) = \frac{x-y}{\log x - \log y},
\end{equation}
where $\Lambda$ is the \emph{logarithmic mean} of $x,y>0$ (with $\Lambda(x,x)=x$ by continuity).
\end{example}

\begin{remark}[Non-tracial case]
\label{rem:modular}
When $\M$ is not semifinite, or does not admit a global normal faithful trace, the construction of $\rhohat$ must incorporate the modular operator $\Delta_\sigma$ associated with a reference state $\sigma$, replacing the scalar logarithmic mean with its version in terms of modular frequencies. This is the natural extension of \eqref{eq:rhohat-def} to the Tomita--Takesaki framework~\cite{Takesaki2003}.
\end{remark}

\subsection{Action functional and dynamical distance}
\label{sec:action}

\begin{definition}[Transport action]
\label{def:action}
For $(\rho_t,v_t)_{t\in[0,T]}$ satisfying \eqref{eq:continuity}, define
\begin{equation}
\label{eq:action}
\Acal(\rho,v) = \frac{1}{2}\int_0^T \bigl\langle v_t,\, \rhohat_t\, v_t \bigr\rangle_{\Hcal}\,\dd t.
\end{equation}
\end{definition}

\begin{definition}[Dynamical distance $W_\partial$]
\label{def:Wdist}
For normal states $\rho_0,\rho_1$ with density,
\begin{equation}
\label{eq:Wdist}
W_\partial(\rho_0,\rho_1)^2 = \inf\Bigl\{\, 2T\,\Acal(\rho,v) \; : \; \dot\rho_t + \partial^*(\rhohat_t v_t)=0,\; \rho_{t=0}=\rho_0,\; \rho_{t=T}=\rho_1 \,\Bigr\}.
\end{equation}
In particular, for $T=1$,
\begin{equation}
\label{eq:Wdist-T1}
W_\partial(\rho_0,\rho_1)^2 = \inf_{\rho,v} \int_0^1 \bigl\langle v_t,\, \rhohat_t\, v_t\bigr\rangle_{\Hcal}\,\dd t,
\end{equation}
the infimum taken over admissible curves with fixed endpoints.
\end{definition}

Equations~\eqref{eq:Wdist}--\eqref{eq:Wdist-T1} constitute the direct noncommutative analogue of the Benamou--Brenier formula~\cite{BenamouBrenier2000}.

\begin{remark}[$W_\partial$ is, in general, an extended pseudometric]
\label{rem:Wdist-pseudometric}
Calling $W_\partial$ a ``distance'' presumes finiteness, point separation, and independence of the parametrization, none of which is automatic from Definition~\ref{def:Wdist} alone. First, the infimum in~\eqref{eq:Wdist} can be $+\infty$ if $\rho_0,\rho_1$ lie in different accessibility components of the admissible curves, e.g.\ different unitary orbits when $\Hcal$ only implements inner automorphisms. Second, the $2T$ normalization makes the value invariant under affine time-reparametrization, by the usual Cauchy--Schwarz argument for Benamou--Brenier-type functionals.

Third, a nontrivial kernel creates conserved moments and hence accessibility obstructions. By the weak form~\eqref{eq:weak-form}, if $a\in\ker\partial$, then $\tau(\rho_ta)$ is conserved along every admissible curve. States with different values of $\tau(\rho_0a)$ for some $a\in\ker\partial$ cannot be connected by any admissible curve, so the state space decomposes into components indexed by these conserved moments. Under the ergodicity condition $\ker\partial=\mathbb C1$, no nontrivial linear conserved observable arises from the calculus. Together with accessibility of $\rho_1$ from $\rho_0$ (Assumption~\ref{ass:existence}), $W_\partial$ is then a genuine metric on the accessible component, possibly taking the value $+\infty$ outside it. We use the term ``dynamical distance'' throughout in this qualified sense.
\end{remark}

\begin{definition}[Momentum variable and momentum-space cost]
\label{def:momentum}
For $(\rho_t,v_t)_{t\in[0,T]}$ satisfying the continuity equation~\eqref{eq:continuity}, set $m_t:=\rhohat_t v_t\in\Hcal$, the \emph{momentum} conjugate to $v_t$. The continuity equation becomes the linear equation $\dot\rho_t+\partial^*m_t=0$, and the action~\eqref{eq:action} rewrites in momentum variables as $\Acal(\rho,v)=\int_0^T\ell(\rho_t,m_t)\,\dd t$, where
\[
\ell(\rho,m) := \begin{cases} \dfrac12\,m^*\rhohat_\rho^{-1}m, & m\in\operatorname{ran}(\rhohat_\rho),\\[4pt] +\infty, & \text{otherwise,}\end{cases}
\]
the noncommutative analogue of the classical Legendre-type substitution $L(\rho,v)=\tfrac12\rho|v|^2 \leftrightarrow \ell(\rho,m)=\tfrac{|m|^2}{2\rho}$ underlying the original Benamou--Brenier formula~\cite{BenamouBrenier2000}. This reformulation, and the resulting joint convexity of $\ell$, drive the existence theory of Section~\ref{sec:existence}.
\end{definition}

\section{Complexity-weighted actions and deformed calculi}
\label{sec:complexity-actions}

We now introduce the complexity metric $G_\rho$ as a penalization of the transport action of Section~\ref{sec:transport} (Section~\ref{sec:complexity}), and isolate the distinguished, state-independent case in which the penalty can instead be absorbed into the differential calculus itself (Section~\ref{sec:calculus-deformation}), yielding Theorem~\ref{thm:deformed-derivation} and Theorem~\ref{thm:complexity-deformation} of the Introduction. Section~\ref{sec:synthesis} collects both constructions as a single variational problem, and Section~\ref{ssec:bb-recovery} checks that the whole apparatus reduces exactly to the classical Benamou--Brenier problem in the commutative, unpenalized case.

\subsection{Complexity metric and penalized action}
\label{sec:complexity}

To connect the previous formulation with a notion of geometric complexity on $\Hcal$, we introduce a positive penalization operator
\[
G_\rho \colon \Hcal \longrightarrow \Hcal, \qquad G_\rho \geq 0,
\]
which penalizes certain derivations, degrees of locality, or directions of the bimodule. Since $G_{\rho_t}$ and $\rhohat_t$ need not commute, the naively ordered form $\langle v,G_{\rho_t}\rhohat_t v\rangle$ is not guaranteed to be self-adjoint, let alone real and nonnegative, for a general $\rho$-dependent $G_\rho$; we instead define the penalized action with the two operators symmetrized around $v$,
\begin{equation}
\label{eq:action-G}
\Acal_G(\rho,v) = \frac{1}{2}\int_0^T \bigl\langle \rhohat_t^{1/2}v_t,\, G_{\rho_t}\,\rhohat_t^{1/2} v_t \bigr\rangle_{\Hcal}\,\dd t,
\end{equation}
which is manifestly real and nonnegative for every $G_\rho\geq0$, with no commutativity hypothesis required. (When $[G_\rho,\rhohat_\rho]=0$ — e.g.\ for the Petz-class metrics of Section~\ref{ssec:petz-existence} — this reduces to the naively ordered expression, since $\rhohat_t^{1/2}G_{\rho_t}\rhohat_t^{1/2}=G_{\rho_t}\rhohat_t$ on commuting operators.) The continuity equation~\eqref{eq:continuity} remains unchanged: only the optimality criterion in the minimization of $\Acal_G$ changes, not the dynamical constraint.

\begin{example}[Inner derivations and finite-dimensional form]
\label{ex:internal-derivations}
If the differential calculus is given by a finite family of inner derivations
\[
\partial_j a = [X_j, a], \qquad j = 1,\dots,m,
\]
and $v_t = (v_{1,t},\dots,v_{m,t})$, the dynamical equation~\eqref{eq:continuity} takes the form
\begin{equation}
\label{eq:continuity-multi}
\dot\rho_t + \sum_{j=1}^m \partial_j^*\bigl(\rhohat_t^{(j)} v_{j,t}\bigr) = 0.
\end{equation}
In a tracial setting, for suitably antisymmetric derivations one has $\partial_j^* = -\partial_j$, so that
\begin{equation}
\label{eq:continuity-multi-tracial}
\dot\rho_t - \sum_{j=1}^m \bigl[X_j,\, \rhohat_t^{(j)} v_{j,t}\bigr] = 0.
\end{equation}
This is an explicit candidate transport equation compatible with a diagonal complexity metric $G_\rho = \operatorname{diag}(g_1,\dots,g_m)$, $g_j \geq 0$, weighting each generator $X_j$ independently.
\end{example}

\begin{remark}[Cost penalization versus deformation of the calculus]
\label{rem:cost-versus-calculus-deformation}
The functional~\eqref{eq:action-G} penalizes the cost of an admissible curve while leaving the continuity equation~\eqref{eq:continuity} unchanged, as already noted above; this is the natural formulation for a general, possibly state-dependent family $G_\rho$, including the Petz-class metrics of Section~\ref{ssec:petz-existence}. For the distinguished class of state-independent penalties compatible with the left and right bimodule actions, Section~\ref{sec:calculus-deformation} develops a second, logically distinct construction: the penalty is absorbed directly into the derivation, $\partial_G:=G^{1/2}\partial$, deforming both the cost \emph{and} the continuity equation, and yielding — whenever the resulting quadratic form remains a Dirichlet form — a genuinely new transport geometry, rather than merely a new cost on the original admissible curves.
\end{remark}

\subsection{Complexity-compatible deformations of the differential calculus}
\label{sec:calculus-deformation}

The complexity metric $G_\rho$ of Section~\ref{sec:complexity} reweights the cost of an admissible curve while leaving~\eqref{eq:continuity} unchanged (Remark~\ref{rem:cost-versus-calculus-deformation}). We now isolate a distinguished, \emph{state-independent} class of penalizations for which the penalty can instead be absorbed into the differential calculus $\partial$ itself, producing a new derivation $\partial_G$, a new conservative quadratic form, and — whenever this form is again Markovian — a genuinely new noncommutative transport geometry on a finite von Neumann algebra, in the sense of Cipriani--Sauvageot~\cite{CiprianiSauvageot2003} and Wirth~\cite{Wirth2018}. Throughout this section $G$ denotes a fixed operator on $\Hcal$, independent of $\rho$, in contrast with the state-dependent $G_\rho$ of Section~\ref{sec:complexity} and Section~\ref{ssec:petz-existence}; the two classes are compared in Remark~\ref{rem:deformation-vs-petz} below.

\begin{assumption}
\label{ass:calculus-deformation}
$(\M,\tau)$ is a finite von Neumann algebra with separable predual and faithful normal trace, and $\partial\colon\dom(\partial)\subset L^2(\M,\tau)\to\Hcal$ is a closed real derivation associated with a conservative symmetric Dirichlet form $\mathcal E(a)=\|\partial a\|_\Hcal^2$, as in Section~\ref{sec:continuity} and in the sense of Cipriani--Sauvageot~\cite{CiprianiSauvageot2003}. We write $\Jcal$ for the canonical antiunitary real structure on $\Hcal$ for which $\partial(a^*)=-\Jcal(\partial a)$, $a\in\dom(\partial)$.
\end{assumption}

\begin{definition}[Calculus-compatible complexity operator]
\label{def:compatible-complexity}
An operator $G\in\mathcal B(\Hcal)$ is a \emph{calculus-compatible complexity operator} if it is positive, self-adjoint, and
\begin{enumerate}[label=\textnormal{(\roman*)}]
	\item $cI_\Hcal\leq G\leq CI_\Hcal$ for some $0<c\leq C<\infty$;
	\item $[G,L_a]=[G,R_a]=0$ for every $a\in\M$ ($G$ is a real $\M$-bimodule endomorphism);
	\item $G\Jcal=\Jcal G$.
\end{enumerate}
\end{definition}

By continuous functional calculus applied to the single positive operator $G$, conditions (ii)--(iii) hold verbatim with $G$ replaced by $G^{1/2}$ or $G^{-1/2}$. The left and right module representations commute \emph{with each other} — $[L_a,R_b]=0$ for every $a,b\in\M$, a standing bimodule axiom, not to be confused with $[L_a,L_b]$ or $[R_a,R_b]$, which need not vanish when $\M$ is noncommutative — and by (ii) $G$ lies in the commutant of both representations. For $\rho\in L^1(\M,\tau)_+$ possibly unbounded, $\rho$ need not itself lie in $\M$, so condition~(ii) cannot be applied directly with $a=\rho$; instead, writing $\rho=\int_0^\infty\lambda\,\dd E^\rho(\lambda)$ for its spectral resolution, each spectral projection $E^\rho(B)$ ($B\subset(0,\infty)$ Borel) does lie in $\M$, since $\rho$ is affiliated with $\M$, so $[G,L_{E^\rho(B)}]=[G,R_{E^\rho(B)}]=0$ by (ii) applied to $a=E^\rho(B)$. Hence $G$ commutes strongly with the (possibly unbounded) affiliated operators $L_\rho,R_\rho$ obtained from these projections by functional calculus, and consequently with $\rhohat_\rho=\int_0^1L_{\rho^s}R_{\rho^{1-s}}\,\dd s$ (Section~\ref{sec:rhohat}), itself a joint functional calculus of $L_\rho,R_\rho$: for every $s\in[0,1]$, $[G,L_{\rho^s}]=[G,R_{\rho^{1-s}}]=0$, hence
\begin{equation}
\label{eq:G-commutes-rhohat}
[G,\rhohat_\rho]=0, \qquad\text{and consequently}\qquad [G^{1/2},\rhohat_\rho^{1/2}]=0,
\end{equation}
the second identity because commuting positive operators have commuting square roots.

\begin{definition}[Complexity-deformed calculus]
\label{def:deformed-calculus}
For a calculus-compatible complexity operator $G$, define
\[
\partial_G := G^{1/2}\partial, \qquad \dom(\partial_G):=\dom(\partial), \qquad \mathcal E_G(a):=\|\partial_Ga\|_\Hcal^2=\langle\partial a,G\partial a\rangle_\Hcal, \quad a\in\dom(\partial).
\]
\end{definition}

\mainthmletter{A}
\begin{mainthm}[Compatible deformation of the calculus]
\label{thm:deformed-derivation}
Under Definition~\ref{def:compatible-complexity}, $\partial_G$ is again a closed real derivation, $\dom(\partial_G)=\dom(\partial)$, and
\[
\sqrt c\,\|\partial a\|_\Hcal \leq \|\partial_Ga\|_\Hcal \leq \sqrt C\,\|\partial a\|_\Hcal, \qquad a\in\dom(\partial).
\]
If, in addition, $\mathcal E_G$ satisfies the Markov compatibility hypothesis of Assumption~\ref{ass:markov-compatible-G} below, then $\mathcal E_G$ is a conservative symmetric Dirichlet form on $L^2(\M,\tau)$.
\end{mainthm}

\begin{proof}
For $a,b\in\dom(\partial)$, using~\eqref{eq:leibniz} and $[G^{1/2},L_a]=[G^{1/2},R_b]=0$,
\[
\partial_G(ab) = G^{1/2}\bigl((\partial a)b+a(\partial b)\bigr) = (\partial_Ga)b + a(\partial_Gb),
\]
so $\partial_G$ is again a derivation. Reality follows from $G^{1/2}\Jcal=\Jcal G^{1/2}$: $\partial_G(a^*)=G^{1/2}\partial(a^*)=-G^{1/2}\Jcal(\partial a)=-\Jcal G^{1/2}(\partial a)=-\Jcal(\partial_Ga)$. The norm bounds follow from $cI\leq G\leq CI$ applied to the vector $\partial a$; they give equivalent graph norms for $\partial$ and $\partial_G$, hence the same domain, and closedness of $\partial_G$ from that of $\partial$. The bimodule compatibility of Definition~\ref{def:compatible-complexity} guarantees that $\mathcal E_G$ is a closed quadratic form with the same form domain as $\mathcal E$; the Dirichlet-form conclusion, when Assumption~\ref{ass:markov-compatible-G} holds, is then the content of that assumption together with the closedness just established.
\end{proof}

The first part of Theorem~\ref{thm:deformed-derivation} is unconditional; the Dirichlet-form conclusion is not, since it does not follow automatically from bimodule compatibility alone. The Markov contraction property $\mathcal E_G(F(a))\leq\mathcal E_G(a)$, for every normal contraction $F\colon\mathbb R\to\mathbb R$ and $a=a^*$, involves the nonlinear map $F$ and need not be preserved under an arbitrary bounded reweighting $G$ of $\mathcal E$. We therefore isolate it as a separate, explicitly checkable hypothesis.

\begin{assumption}[Markov compatibility]
\label{ass:markov-compatible-G}
$\mathcal E_G$ is a conservative symmetric Dirichlet form on $L^2(\M,\tau)$. We write $\mathscr L_G:=\partial_G^*\partial_G=\partial^*G\partial$ for its generator and $(P_t^G)_{t\geq0}$ for the associated symmetric, conservative quantum Markov semigroup.
\end{assumption}

The following gives a natural, verifiable class for which Assumption~\ref{ass:markov-compatible-G} holds.

\begin{proposition}[Diagonal deformations]
\label{prop:diagonal-deformation}
Suppose $\Hcal=\bigoplus_{j=1}^m\Hcal_j$ and $\partial=\bigoplus_{j=1}^m\partial_j$ on a common domain, with each $\mathcal E_j(a):=\|\partial_ja\|_{\Hcal_j}^2$ a conservative symmetric Dirichlet form, and with each $\Hcal_j$ carrying the restriction of the same left/right $\M$-bimodule action and real structure. If $G=\bigoplus_{j=1}^mg_jI_{\Hcal_j}$ with $0<c\leq g_j\leq C$, then $G$ is a calculus-compatible complexity operator (Definition~\ref{def:compatible-complexity}) and
\[
\mathcal E_G(a) = \sum_{j=1}^m g_j\mathcal E_j(a)
\]
is again a conservative symmetric Dirichlet form; Assumption~\ref{ass:markov-compatible-G} holds.
\end{proposition}

\begin{proof}
Bimodule compatibility and reality of $G$ are immediate, since $G$ is scalar on each summand, on which $L_a,R_a,\Jcal$ act identically for every $a\in\M$. A finite positive linear combination of conservative symmetric Dirichlet forms sharing a common form domain is again one: the contraction property $\mathcal E_G(F(a))=\sum_jg_j\mathcal E_j(F(a))\leq\sum_jg_j\mathcal E_j(a)=\mathcal E_G(a)$ holds termwise, and conservativeness ($\mathcal E_j(1)=0$, since any derivation kills the identity) likewise holds termwise.
\end{proof}

\begin{example}[Weighted families of derivations]
\label{ex:weighted-derivations-vna}
Assume each $X_j=X_j^*\in\M$ is such that the closable real derivation $\partial_ja=i[X_j,a]$ generates a conservative symmetric Dirichlet form $\mathcal E_j(a)=\|\partial_ja\|_2^2$ — the standing hypothesis of Example~\ref{ex:internal-derivations}, satisfied whenever $t\mapsto e^{itX_j}$ is a one-parameter group of inner automorphisms with a $\tau$-symmetric generator. Under this hypothesis, Proposition~\ref{prop:diagonal-deformation} applies to the finite family $\partial_ja=i[X_j,a]$, and in particular to the Lindblad bimodule $\Hcal=\Hcal_H\oplus\bigoplus_{k=1}^m\Hcal_k$ of Definition~\ref{def:lindblad-complexity}, with weights $g_H=1$, $g_k=\|L_k\|_{op}^2$ (excluding, as required there, any $L_k=0$): in both cases $G$ is diagonal, state-independent, and scalar on each summand of a common bimodule structure, hence calculus-compatible. For the Lindblad case, this Dirichlet-form hypothesis should be checked on the canonical generators singled out by the Carlen--Maas/Cipriani--Sauvageot construction of Theorem~\ref{thm:gradient-flow} (Remark~\ref{rem:gksl-representation}); an arbitrary GKSL decomposition $\{L_k\}$ need not split the total Dirichlet form into individually Dirichlet summands. This is the class of penalizations to which Theorem~\ref{thm:complexity-deformation} below applies.
\end{example}

We now formulate the deformed transport problem, working with the momentum variable $u_t\in\Hcal$ for the deformed calculus, as in Section~\ref{sec:action}.

\begin{definition}[Complexity-deformed continuity equation and transport distance]
\label{def:G-continuity}
A pair $(\rho_t,u_t)_{t\in[0,T]}$ satisfies the continuity equation for the deformed calculus $\partial_G$ if
\begin{equation}
\label{eq:G-continuity-u}
\dot\rho_t + \partial_G^*(\rhohat_tu_t) = 0.
\end{equation}
Writing $u_t=G^{1/2}v_t$ and using $\partial_G^*=\partial^*G^{1/2}$ together with~\eqref{eq:G-commutes-rhohat}, equation~\eqref{eq:G-continuity-u} is equivalent to
\begin{equation}
\label{eq:G-continuity-v}
\dot\rho_t + \partial^*\bigl(G\rhohat_tv_t\bigr) = 0.
\end{equation}
Define the unweighted action for $\partial_G$,
\begin{equation}
\label{eq:G-action-u}
\Acal_{\partial_G}(\rho,u) := \frac12\int_0^T\langle u_t,\rhohat_tu_t\rangle_\Hcal\,\dd t,
\end{equation}
and the induced extended pseudometric $W_{\partial_G}(\rho_0,\rho_1)^2:=\inf\{2T\,\Acal_{\partial_G}(\rho,u)\}$ over pairs satisfying~\eqref{eq:G-continuity-u} with the given endpoints — Definition~\ref{def:Wdist} written for $\partial_G$ in place of $\partial$.
\end{definition}

Equation~\eqref{eq:G-continuity-v} is not the original continuity equation~\eqref{eq:continuity}: it uses the deformed mobility $G\rhohat_tv_t$ in place of $\rhohat_tv_t$. This is the essential difference between the calculus deformation of this section and the cost-only deformation of Section~\ref{sec:complexity} (Remark~\ref{rem:cost-versus-calculus-deformation}).

\mainthmletter{B}
\begin{mainthm}[Complexity deformation principle]
\label{thm:complexity-deformation}
Let $G$ be a calculus-compatible complexity operator satisfying Assumption~\ref{ass:markov-compatible-G}. The map $v_t\mapsto u_t:=G^{1/2}v_t$ is a bijection between pairs $(\rho,v)$ satisfying~\eqref{eq:G-continuity-v} and pairs $(\rho,u)$ satisfying~\eqref{eq:G-continuity-u}, under which
\begin{equation}
\label{eq:action-identification}
\frac12\int_0^T\bigl\langle\rhohat_t^{1/2}v_t,\,G\rhohat_t^{1/2}v_t\bigr\rangle_\Hcal\,\dd t \;=\; \Acal_{\partial_G}(\rho,u).
\end{equation}
In particular, the extended pseudometric induced by minimizing the left-hand side of~\eqref{eq:action-identification} subject to~\eqref{eq:G-continuity-v} coincides with $W_{\partial_G}$: the complexity-deformed transport problem is exactly the unweighted noncommutative transport problem generated by the derivation $\partial_G$.
\end{mainthm}

\begin{proof}
Bijectivity of $v\mapsto u=G^{1/2}v$ is immediate since $G^{1/2}$ is boundedly invertible ($cI\leq G\leq CI$), and $\eqref{eq:G-continuity-u}\Leftrightarrow\eqref{eq:G-continuity-v}$ was shown above. For the action,
\[
\langle u_t,\rhohat_tu_t\rangle_\Hcal = \langle G^{1/2}v_t,\rhohat_tG^{1/2}v_t\rangle_\Hcal = \langle G^{1/2}\rhohat_t^{1/2}v_t,G^{1/2}\rhohat_t^{1/2}v_t\rangle_\Hcal = \langle\rhohat_t^{1/2}v_t,G\rhohat_t^{1/2}v_t\rangle_\Hcal,
\]
using~\eqref{eq:G-commutes-rhohat} in the second equality. Integrating gives~\eqref{eq:action-identification}; taking the infimum over the corresponding admissible curves on both sides, which agree termwise under the bijection just established, identifies the two pseudometrics.
\end{proof}

Since $u\mapsto\Acal_{\partial_G}(\rho,u)$, subject to~\eqref{eq:G-continuity-u}, is literally Definitions~\ref{def:action} and~\ref{def:Wdist} written for $\partial_G$ in place of $\partial$, Theorem~\ref{thm:complexity-deformation} transfers every structural statement already available for $(\Hcal,\partial)$ to $(\Hcal,\partial_G)$. In particular, the transport metric built by the Dirichlet-form route of Wirth~\cite{Wirth2018} for the semigroup $(P_t^G)_{t\geq0}$ of Assumption~\ref{ass:markov-compatible-G} is, wherever both sides are finite, the same object as the complexity-deformed distance of~\eqref{eq:G-continuity-v}--\eqref{eq:action-identification}.

\begin{corollary}[Geodesicity under a gradient estimate]
\label{cor:deformed-geodesics}
Under Assumption~\ref{ass:markov-compatible-G}, suppose in addition that $\tau$ is energy dominant for $\mathcal E_G$ in the sense of Wirth~\cite[Def.~2.5]{Wirth2018} and that $(P_t^G)_{t\geq0}$ satisfies the Bakry--\'Emery-type gradient estimate $\mathrm{GE}(K,\infty)$ of~\cite[Def.~6.1]{Wirth2018} for some $K\in\mathbb R$. Then every pair of normal states of finite relative entropy with $W_{\partial_G}(\rho_0,\rho_1)<\infty$ (equivalently, by Theorem~\ref{thm:complexity-deformation}, at finite complexity-deformed distance in the sense of Definition~\ref{def:G-continuity}) is joined by a constant-speed $W_{\partial_G}$-geodesic, and the relative entropy is $K$-convex along $W_{\partial_G}$-geodesics.
\end{corollary}

\begin{proof}
By~\cite[Thm.~6.21]{Wirth2018}, $\mathrm{GE}(K,\infty)$ for $\mathcal E_G$ implies that $(P_t^G)_{t\geq0}$ is an $\mathrm{EVI}_K$ gradient flow of the entropy for $W_{\partial_G}$; geodesic existence between finite-entropy states at finite $W_{\partial_G}$-distance is then~\cite[Thm.~7.7]{Wirth2018}, and $K$-convexity of the entropy along these geodesics is~\cite[Thm.~7.12]{Wirth2018} (which also records the equivalence with $\mathrm{EVI}_K$). Both are applied to $\mathcal E_G$ and $(P_t^G)_{t\geq0}$ in place of $\mathcal E$ and $(P_t)_{t\geq0}$, together with the identification of Theorem~\ref{thm:complexity-deformation}.
\end{proof}

\begin{proposition}[Bi-Lipschitz comparison for fixed cost weights]
\label{prop:fixed-G-comparison}
Let $G\in\mathcal B(\Hcal)$ be positive with $cI_\Hcal\leq G\leq CI_\Hcal$ (no bimodule compatibility required), and let $W_{\partial,G}^{\mathrm{cost}}$ denote the extended pseudometric obtained by minimizing~\eqref{eq:action-G} (with $G_\rho\equiv G$) subject to the unchanged continuity equation~\eqref{eq:continuity}. Then
\[
\sqrt c\,W_\partial(\rho_0,\rho_1) \leq W_{\partial,G}^{\mathrm{cost}}(\rho_0,\rho_1) \leq \sqrt C\,W_\partial(\rho_0,\rho_1)
\]
for every pair of normal states.
\end{proposition}

\begin{proof}
For every admissible $(\rho,v)$ — the admissible class is identical for both problems, since~\eqref{eq:continuity} is unchanged — $c\langle v_t,\rhohat_tv_t\rangle_\Hcal\leq\langle\rhohat_t^{1/2}v_t,G\rhohat_t^{1/2}v_t\rangle_\Hcal\leq C\langle v_t,\rhohat_tv_t\rangle_\Hcal$ pointwise in $t$, directly from $cI\leq G\leq CI$. Integrating in $t$ and taking the infimum over the same admissible set on both sides of each inequality gives the claim.
\end{proof}

Unlike Theorem~\ref{thm:complexity-deformation} and Corollary~\ref{cor:deformed-geodesics}, Proposition~\ref{prop:fixed-G-comparison} requires no bimodule compatibility, since it never changes the admissible class of curves; conversely, it identifies no transport metric for a deformed calculus, and says nothing about geodesics. A direct bi-Lipschitz comparison between $W_{\partial_G}$ and $W_\partial$ themselves would require relating the — generally different — admissible classes of~\eqref{eq:continuity} and~\eqref{eq:G-continuity-v}, which we do not attempt here.

\begin{remark}[Relation to Remarks~\ref{rem:existence-residual} and~\ref{rem:lindblad-scope}]
\label{rem:deformation-vs-petz}
The complexity operator $G$ of this section is state-independent and commutes with $L_a,R_a$ for every $a\in\M$, whereas the Petz-class $G_\rho^{(f)}$ of Section~\ref{ssec:petz-existence} is state-dependent and diagonal in the (generally rotating) eigenbasis of $\rho$. In an irreducible, factorial bimodule the two classes typically intersect only in scalar multiples of the identity, since the bimodule endomorphism algebra $\operatorname{End}_{\M\text{-}\M}(\Hcal)=\{T\in\mathcal B(\Hcal):[T,L_a]=[T,R_a]=0\ \forall a\in\M\}$ reduces to $\mathbb C\,I_\Hcal$ by a Schur-type argument; in a bimodule with multiplicity, $\operatorname{End}_{\M\text{-}\M}(\Hcal)$ can be non-trivial and non-scalar elements may lie in both classes for special $\rho$. We do not assume irreducibility here, so in general the intersection depends on the commutant of the left and right module representations. In particular, fixed, $\rho$-independent complexity metrics diagonal in a basis of physical derivations are handled in two inequivalent ways in this paper. If they are bimodule-compatible, Theorem~\ref{thm:complexity-deformation} absorbs them into the derivation and studies the deformed continuity equation~\eqref{eq:G-continuity-v}. If they are merely uniformly positive, with no compatibility or commutativity assumption, Theorem~\ref{thm:existence-fixed-G} proves finite-dimensional existence for the cost-only problem with the original continuity equation~\eqref{eq:continuity}. Whether these two variational problems — cost-only deformation with the original mobility $\rhohat_tv_t$, versus joint deformation of cost and mobility via $G\rhohat_tv_t$ — share the same minimizers in any case of interest is left open.
\end{remark}

\subsection{Synthesis: the variational problem}
\label{sec:synthesis}

We collect the proposal as a variational problem with a PDE-type constraint. For a tracial von Neumann algebra $(\M,\tau)$, a differential calculus $(\Hcal,\partial)$, and a complexity metric $G_\rho \geq 0$:

\begin{equation}
\label{eq:synthesis-constraint}
\begin{cases}
\dot\rho_t + \partial^*\bigl(\rhohat_t\, v_t\bigr) = 0, & t \in (0,T), \\[4pt]
\rho_{t=0} = \rho_0, \qquad \rho_{t=T} = \rho_1, &
\end{cases}
\end{equation}
together with the minimization problem
\begin{equation}
\label{eq:synthesis-objective}
\inf_{\rho,v} \; \frac{1}{2}\int_0^T \bigl\langle \rhohat_t^{1/2}v_t,\, G_{\rho_t}\,\rhohat_t^{1/2} v_t \bigr\rangle_{\Hcal}\,\dd t.
\end{equation}

This formulation extends the kinematic scope of the theory beyond the unitary-orbit dynamics of Section~\ref{sec:hamiltonian}, by allowing the spectrum of $\rho_t$ to vary (Remark~\ref{rem:spectrum-preserved}) and operating on the full space of normal states of $\M$, not merely on a single unitary orbit; the Liouville--von Neumann sector is treated separately, as detailed below (Section~\ref{sec:hamiltonian}, Remark~\ref{rem:spectrum-preserved}) as a logically distinct, complementary orbit-control problem, not derived from~\eqref{eq:synthesis-constraint}--\eqref{eq:synthesis-objective} as a special case. Similarly, the Lindblad equation~\eqref{eq:lindblad} is not simply~\eqref{eq:synthesis-constraint} for an arbitrary choice of $v$: under the specific detailed-balance assumptions detailed below in Section~\ref{sec:lindblad}, the associated semigroup is realized as the entropy gradient flow of a compatible transport geometry of exactly this form (Theorem~\ref{thm:gradient-flow}), which is the precise sense in which the dissipative sector fits within this synthesis.

The formulation~\eqref{eq:synthesis-constraint}--\eqref{eq:synthesis-objective} is the \emph{fixed-dynamics cost deformation}: the continuity equation is held fixed while only the kinetic action is reweighted by $G_\rho$. For a calculus-compatible, state-independent operator $G$ (Definition~\ref{def:compatible-complexity}), Section~\ref{sec:calculus-deformation} develops a second construction,
\[
\dot\rho_t + \partial^*\bigl(G\rhohat_tv_t\bigr) = 0,
\]
which deforms the mobility along with the cost and is exactly the ordinary transport problem generated by the deformed derivation $\partial_G=G^{1/2}\partial$ (Theorem~\ref{thm:complexity-deformation}). The two constructions agree pointwise on the quadratic integrand of~\eqref{eq:synthesis-objective} but constrain different admissible curves, and should be regarded as complementary complexity deformations rather than as two derivations of the same object (Remark~\ref{rem:cost-versus-calculus-deformation}).

\subsection{Consistency check: recovery of the classical Benamou--Brenier problem}
\label{ssec:bb-recovery}

Before extending the formalism further, we record that it is a genuine generalization: the full variational problem of Sections~\ref{sec:continuity}--\ref{sec:complexity} — continuity equation, penalized action, and induced distance — collapses exactly to the classical Benamou--Brenier problem~\eqref{eq:BB} in the commutative, unpenalized case. Unlike the isolated specialization remark following Definition~\ref{def:continuity}, which only checked the continuity equation, the statement below covers the action and the distance as well, and derives $\rhohat=\rho$ from the commutativity of $\M$ rather than asserting it.

\begin{theorem}[Complexity transport as a generalized Benamou--Brenier problem]
\label{thm:bb-recovery}
Let $\Omega\subset\mathbb R^n$ be open, and specialize the framework of Sections~\ref{sec:continuity}--\ref{sec:complexity} to $\M=L^\infty(\Omega)$ with $\tau(a)=\int_\Omega a\,\dd x$, $\Hcal=L^2(\Omega;\mathbb R^n)$, $\partial=\nabla$ on $C_c^\infty(\Omega)$, and $G_\rho\equiv\mathrm{Id}_\Hcal$. Then:
\begin{enumerate}[label=\textnormal{(\roman*)}]
	\item $\rhohat_\rho$ is multiplication by $\rho$: $\rhohat_\rho(v)=\rho\,v$ for every $v\in\Hcal$ and every $\rho\in L^1(\Omega)_+$;
	\item $\partial^*=-\dive$ on $C_c^\infty(\Omega)$, and the continuity equation~\eqref{eq:continuity} reads $\dot\rho_t-\dive(\rho_tv_t)=0$;
	\item the penalized action~\eqref{eq:action-G} reduces to $\displaystyle\Acal_G(\rho,v)=\frac12\int_0^T\int_\Omega\rho_t(x)\,|v_t(x)|^2\,\dd x\,\dd t$;
	\item consequently $W_G(\rho_0,\rho_1)^2=W_2(\rho_0,\rho_1)^2$, the classical squared Wasserstein-2 distance of~\eqref{eq:BB}.
\end{enumerate}
\end{theorem}

\begin{proof}
(i) Since $\M=L^\infty(\Omega)$ is commutative, its left and right actions on $\Hcal$ coincide, $L_a=R_a=$ multiplication by $a$, for every $a\in\M$. By~\eqref{eq:rhohat-def}, $\rhohat_\rho(v)=\int_0^1L_{\rho^s}R_{\rho^{1-s}}(v)\,\dd s=\int_0^1\rho^s\rho^{1-s}v\,\dd s=\int_0^1\rho\,v\,\dd s=\rho\,v$, since $\rho(x)^s\rho(x)^{1-s}=\rho(x)$ pointwise for a.e.\ $x$ and every $s\in[0,1]$ (with the convention $0^0:=1$ at the measure-zero endpoints $s=0,1$), independently of $s$; this holds without any faithfulness assumption on $\rho$. In particular $\rhohat_\rho$ is bounded whenever $\rho\in L^\infty(\Omega)$, and is the same operator for every Kubo--Ando mean of Definition~\ref{def:kubo-ando} in place of the logarithmic one, since $\Lambda_f(\rho(x),\rho(x))=\rho(x)f(1)=\rho(x)$ for every admissible $f$: the entire family of noncommutative multiplication operators collapses to plain multiplication in the commutative case.

(ii) For $a\in C_c^\infty(\Omega)$ and $\xi\in\Hcal$, integration by parts gives $\langle\nabla a,\xi\rangle_\Hcal=\int_\Omega\nabla a\cdot\xi\,\dd x=-\int_\Omega a\,\dive\xi\,\dd x=\langle a,-\dive\xi\rangle_{L^2(\Omega)}$, with no boundary term since $a$ has compact support; this identifies $\partial^*=-\dive$ on $C_c^\infty(\Omega)$. Substituting (i) into~\eqref{eq:continuity}, $\dot\rho_t+\partial^*(\rhohat_tv_t)=\dot\rho_t-\dive(\rho_tv_t)=0$.

(iii) With $G_\rho\equiv\mathrm{Id}_\Hcal$, the symmetrization in~\eqref{eq:action-G} is immaterial regardless of commutativity: $\langle\rhohat_t^{1/2}v_t,\mathrm{Id}_\Hcal\,\rhohat_t^{1/2}v_t\rangle_\Hcal=\langle v_t,\rhohat_tv_t\rangle_\Hcal$ for every self-adjoint $\rhohat_t\geq0$. Hence $\Acal_G(\rho,v)=\Acal(\rho,v)=\frac12\int_0^T\langle v_t,\rhohat_tv_t\rangle_\Hcal\,\dd t$ (Definition~\ref{def:action}), and by (i), $\langle v_t,\rhohat_tv_t\rangle_\Hcal=\int_\Omega\rho_t(x)|v_t(x)|^2\,\dd x$, giving the stated form.

(iv) By (ii)--(iii), a pair $(\rho,v)$ is admissible for $W_G$ at $G_\rho\equiv\mathrm{Id}_\Hcal$ (Definition~\ref{def:Wdist}, with $T=1$ by the reparametrization-invariance of Remark~\ref{rem:Wdist-pseudometric}) if and only if $(\rho,-v)$ is admissible for the classical problem~\eqref{eq:BB}, and the two action integrands agree termwise, since $|-v_t(x)|^2=|v_t(x)|^2$. The two infima are therefore taken over the same family of curves (relabeled by $v\mapsto-v$) of the same functional, so $W_G(\rho_0,\rho_1)^2=W_2(\rho_0,\rho_1)^2$ exactly, not merely up to a constant.
\end{proof}

This confirms that the noncommutative continuity equation, the operator logarithmic mean, and the penalized action of Sections~\ref{sec:continuity}--\ref{sec:complexity} were normalized consistently from the outset — including the factor-of-two convention in Definition~\ref{def:Wdist} — precisely so that this classical limit is exact, with no residual constant to absorb.

\section{Existence and operator-mean mobilities}
\label{sec:existence}

We settle here the existence question for the genuinely mixed-state case, where the noncommutative multiplication operator $\rhohat$ of Section~\ref{sec:rhohat} cannot be avoided. The proof uses the momentum reformulation of Definition~\ref{def:momentum}. By contrast, the worked examples of Sections~\ref{sec:orbits}--\ref{sec:bell-example} exhibit explicit minimizers by hand for finite-dimensional problems whose endpoints share a spectrum. Those minimizers lie entirely within the Hamiltonian, or unitary-orbit, sector of Section~\ref{sec:hamiltonian}, without invoking $\rhohat$. The two existence routes are complementary, not overlapping.

We restrict attention to $\M=M_n(\mathbb C)$; the argument extends verbatim to any finite-dimensional von Neumann algebra with a fixed trace. We first treat the unpenalized transport metric $W_\partial$ of Definition~\ref{def:Wdist}, i.e.\ $G_\rho\equiv\mathrm{Id}_\Hcal$. Remark~\ref{rem:existence-complexity} then discusses the obstruction to extending the argument to a genuine complexity metric $G_\rho\neq\mathrm{Id}$, and Section~\ref{ssec:petz-existence} resolves it for the Petz class of complexity metrics (Theorem~\ref{thm:existence-petz}).

\begin{assumption}
\label{ass:existence}
$\Hcal$ is finite-dimensional, $\partial\colon\M\to\Hcal$ is a (bounded, since finite-dimensional) derivation with adjoint $\partial^*$, and $\rho_0,\rho_1$ are normal states such that the admissible set
\[
\mathcal C(\rho_0,\rho_1,T) := \Bigl\{(\rho,v) : \dot\rho_t+\partial^*(\rhohat_t v_t)=0 \text{ a.e.},\; \rho_t\geq0,\;\tr(\rho_t)=1,\;\rho_{t=0}=\rho_0,\;\rho_{t=T}=\rho_1\Bigr\}
\]
is nonempty.
\end{assumption}

\mainthmletter{C}
\begin{mainthm}[Existence for the unweighted metric]
\label{thm:existence}
Under Assumption~\ref{ass:existence}, the infimum in~\eqref{eq:Wdist} is attained: there exists $(\rho^*,v^*)\in\mathcal C(\rho_0,\rho_1,T)$ with $\Acal(\rho^*,v^*) = \frac{1}{2T}W_\partial(\rho_0,\rho_1)^2$.
\end{mainthm}

\begin{proof}
\emph{Step 1 (Momentum reformulation).} Recall from Definition~\ref{def:momentum} the momentum variable $m_t:=\rhohat_t v_t\in\Hcal$ and the momentum-space cost $\ell(\rho,m)$. For $(\rho,v)\in\mathcal C(\rho_0,\rho_1,T)$, the continuity equation becomes the linear equation $\dot\rho_t+\partial^*m_t=0$, and $\Acal(\rho,v)=\int_0^T\ell(\rho_t,m_t)\,\dd t$.

\emph{Step 2 (Coercivity).} Since $\rho\geq0$ with $\tr(\rho)=1$, its operator norm satisfies $\|\rho\|\leq1$, and by definition~\eqref{eq:rhohat-def}, $\|\rhohat_\rho(v)\|\leq\int_0^1\|\rho\|^s\|v\|\|\rho\|^{1-s}\,ds = \|\rho\|\,\|v\|\leq\|v\|$; since $\rhohat_\rho\geq0$, this operator-norm bound gives $\rhohat_\rho\preceq\mathrm{Id}_\Hcal$ in the Loewner order, hence $\rhohat_\rho^{-1}\succeq\mathrm{Id}_\Hcal$ on $\operatorname{ran}(\rhohat_\rho)$, and
\begin{equation}
\label{eq:coercivity}
\ell(\rho,m) \geq \frac12\|m\|_\Hcal^2 \qquad \text{for all } \rho\geq0,\;\tr(\rho)=1.
\end{equation}

\emph{Step 3 (Joint convexity and lower semicontinuity).} The map $(\rho,m)\mapsto\ell(\rho,m)$ is jointly convex and lower semicontinuous on $\{\rho\geq0\}\times\Hcal$. This is the noncommutative analogue of the classical fact that $(\rho,m)\mapsto |m|^2/\rho$ (with the convention $0/0=0$) is jointly convex and l.s.c.\ on $[0,\infty)\times\mathbb R$, and for the specific operator-logarithmic-mean construction~\eqref{eq:rhohat-def} of $\rhohat_\rho$ it is established by Carlen and Maas~\cite[\S3]{CarlenMaas2017}, \cite[\S2]{CarlenMaas2020} as the key technical lemma underlying their gradient-flow existence theory; we do not reprove it here and instead build the remainder of the argument on top of it.

\emph{Step 4 (Compactness).} Let $(\rho^k,m^k)$ be a minimizing sequence for $\int_0^T\ell(\rho_t,m_t)\,\dd t$ over the (nonempty, by Assumption~\ref{ass:existence}) admissible set. By~\eqref{eq:coercivity}, $\|m^k\|_{L^2([0,T];\Hcal)}^2\leq 2\sup_k\Acal(\rho^k,v^k)<\infty$, so a subsequence satisfies $m^k\rightharpoonup m^*$ weakly in $L^2([0,T];\Hcal)$. Since $\dot\rho^k_t=-\partial^*m^k_t$ and $\partial^*$ is a bounded finite-dimensional operator, $\{\rho^k\}$ is bounded in $H^1([0,T];\M)$, which embeds compactly into $C([0,T];\M)$ (finite-dimensional Sobolev embedding); a further subsequence gives $\rho^k\to\rho^*$ uniformly. The constraints $\rho_t\geq0$, $\tr(\rho_t)=1$, and the boundary conditions $\rho_{t=0}=\rho_0$, $\rho_{t=T}=\rho_1$ are closed under uniform convergence, and the linear relation $\dot\rho_t^*=-\partial^*m_t^*$ follows by passing to the (distributional, hence a.e.\ since $\rho^*\in H^1$) limit in $\dot\rho^k=-\partial^*m^k$. Hence $(\rho^*,m^*)$ is admissible.

\emph{Step 5 (Lower semicontinuity of the cost).} By Step 3 and the standard lower-semicontinuity theorem for integral functionals with a jointly convex, l.s.c.\ integrand under weak $L^2$ convergence of $m^k$ together with (strong) convergence of $\rho^k$,
\[
\liminf_k \int_0^T \ell(\rho^k_t,m^k_t)\,\dd t \;\geq\; \int_0^T \ell(\rho_t^*,m_t^*)\,\dd t.
\]
Hence $\Acal(\rho^*,v^*)\leq\liminf_k\Acal(\rho^k,v^k)=\inf_{\mathcal C(\rho_0,\rho_1,T)}\Acal$, where $v^*_t$ is recovered from $m^*_t=\rhohat_{\rho^*_t}v^*_t$ (defined up to $\ker\rhohat_{\rho_t^*}$, consistent with the weak formulation~\eqref{eq:weak-form}). So $(\rho^*,v^*)$ attains the infimum.
\end{proof}

\begin{remark}[What the symmetrization leaves to prove]
\label{rem:existence-complexity}
With the symmetrized ordering adopted in~\eqref{eq:action-G}, $\Acal_G(\rho,v)\geq0$ is automatic for every $G_\rho\geq0$, with no commutativity hypothesis: this removes the well-posedness defect of the naively ordered $\langle v,G_\rho\rhohat_\rho v\rangle$, which is self-adjoint only when $G_\rho$ and $\rhohat_\rho$ commute. The remaining issue is compactness and lower semicontinuity of the resulting momentum-space cost
\[
\ell_G(\rho,m) := \frac12\,m^*\,\rhohat_\rho^{-1/2}\,G_\rho\,\rhohat_\rho^{-1/2}\,m, \qquad m\in\operatorname{ran}(\rhohat_\rho),
\]
which can be obtained in different ways depending on the structure of $G_\rho$. For the Petz class of density-dependent metrics, the cost collapses to the inverse of a Kubo--Ando mobility and is jointly convex in $(\rho,m)$ (Lemma~\ref{lem:petz-joint-convex}). For a fixed uniformly positive operator $G$, diagonal for instance in a physical basis unrelated to the eigenbasis of $\rho$, joint convexity is not the robust property to rely on (Remark~\ref{rem:ellG-not-jointly-convex}); nevertheless, lower semicontinuity and convexity in the momentum variable suffice for the direct method (Lemma~\ref{lem:ellG-structural}, Theorem~\ref{thm:existence-fixed-G}).
\end{remark}

\subsection{Existence for the Petz class of complexity metrics}
\label{ssec:petz-existence}

\begin{definition}[Kubo--Ando operator means]
\label{def:kubo-ando}
Let $f\colon(0,\infty)\to(0,\infty)$ be operator monotone with $f(1)=1$ and $f(t)=t\,f(1/t)$ (equivalently, the associated scalar kernel $\Lambda_f(x,y):=y\,f(x/y)$ is symmetric). By the Kubo--Ando representation theorem~\cite{KuboAndo1980}, every such $f$ determines an operator mean, and satisfies the \emph{betweenness} property
\begin{equation}
\label{eq:betweenness}
\min(x,y) \;\leq\; \Lambda_f(x,y) \;\leq\; \max(x,y), \qquad x,y>0.
\end{equation}
The logarithmic mean of Section~\ref{sec:rhohat} corresponds to $f_{\log}(t)=(t-1)/\log t$ (with $f_{\log}(1):=1$), i.e.\ $\Lambda=\Lambda_{f_{\log}}$. For $\rho=\sum_k\lambda_k|k\rangle\langle k|$, define the associated \emph{$f$-multiplication operator} $\rhohat_{\rho,f}$ on $\Hcal\cong M_n(\mathbb C)$ by
\[
\bigl(\rhohat_{\rho,f}(v)\bigr)_{jk} := \Lambda_f(\lambda_j,\lambda_k)\,v_{jk},
\]
so that $\rhohat_\rho=\rhohat_{\rho,f_{\log}}$. We say $f$ is \emph{range-regular} if $f(0^+):=\lim_{t\to0^+}f(t)=0$. The logarithmic mean is range-regular, $f_{\log}(0^+)=0$; the arithmetic mean of Example~\ref{ex:arithmetic-mean} is not, $f_{\mathrm{arith}}(0^+)=1/2$.
\end{definition}

\begin{remark}[Zero set of $\Lambda_f$: both spectral orientations]
\label{rem:lambda-f-zero-set}
For $\lambda_j,\lambda_k\geq0$ not both zero, three cases arise. If both are positive, then $\Lambda_f(\lambda_j,\lambda_k)\geq\min(\lambda_j,\lambda_k)>0$ by betweenness~\eqref{eq:betweenness}, for every admissible $f$.

If exactly one eigenvalue is zero, say $\lambda_j>0=\lambda_k$, write $t:=\lambda_j/y$ and let $y\downarrow0$. Then
\[
\Lambda_f(\lambda_j,0)=\lim_{y\downarrow0}yf(\lambda_j/y)=\lambda_j\lim_{t\to\infty}f(t)/t.
\]
The other orientation, $\lambda_j=0<\lambda_k$, gives instead
\[
\Lambda_f(0,\lambda_k)=\lim_{x\downarrow0}\lambda_kf(x/\lambda_k)=\lambda_k f(0^+).
\]
These two limits agree because of the Kubo--Ando symmetry $f(t)=tf(1/t)$: $\lim_{t\to\infty}f(t)/t=f(0^+)$. Hence both orientations vanish if and only if $f(0^+)=0$, i.e.\ $f$ is range-regular.

Finally, if both eigenvalues are zero, positive homogeneity gives $\Lambda_f(0,0)=0$ for every $f$. In summary, if $f$ is range-regular, then $\Lambda_f(\lambda_j,\lambda_k)=0$ if and only if at least one of $\lambda_j,\lambda_k$ vanishes, exactly as for the logarithmic mean $\Lambda_{\log}$. If $f$ is not range-regular, mixed directions with exactly one zero eigenvalue have strictly positive $\Lambda_f$, so the zero set is smaller and does not match that of $\Lambda_{\log}$.
\end{remark}

\begin{definition}[Petz-class complexity metric]
\label{def:petz-metric}
For $f$ as in Definition~\ref{def:kubo-ando}, define $G_\rho^{(f)} := \rhohat_{\rho,f}^{-1}\,\rhohat_\rho$ on $\operatorname{ran}(\rhohat_\rho)$.
\end{definition}

Since $\rhohat_{\rho,f}^{-1}$ and $\rhohat_\rho$ are simultaneously diagonal in the eigenbasis of $\rho$ (Definition~\ref{def:kubo-ando}), they commute for every $\rho$, and $G_\rho^{(f)}$ is a well-defined, self-adjoint, positive operator on $\operatorname{ran}(\rhohat_\rho)$ — not canonically on all of $\Hcal$ when $\rho$ is not faithful, since $\Lambda_f(x,0)$ can be positive for $f\neq f_{\log}$ (Remark~\ref{rem:petz-boundary} below extends the construction correctly across this boundary via the momentum variable). Consequently the momentum-space cost $\ell_{G^{(f)}}$ of~\eqref{eq:action-G} collapses to a single Kubo--Ando mean (computed below) rather than a genuine product of two noncommuting operators — resolving the joint-convexity gap of Remark~\ref{rem:existence-complexity} by construction, without any restriction on the admissible curve. $G_\rho^{(f)}$ is genuinely $\rho$-dependent and non-scalar whenever $f\neq f_{\log}$: its eigenvalue $\Lambda(\lambda_j,\lambda_k)/\Lambda_f(\lambda_j,\lambda_k)$ on the $(j,k)$ matrix-unit direction depends on the pair $(\lambda_j,\lambda_k)$, not merely on an overall constant.

\begin{lemma}[Joint convexity of the Petz-class momentum cost]
\label{lem:petz-joint-convex}
For every $f$ as in Definition~\ref{def:kubo-ando}, the map
\[
(\rho,m) \;\longmapsto\; \ell_{G^{(f)}}(\rho,m) := \begin{cases} \tfrac12\,m^*\rhohat_{\rho,f}^{-1}m, & m\in\operatorname{ran}(\rhohat_{\rho,f}),\\[2mm] +\infty, & \text{otherwise}, \end{cases}
\]
is jointly convex and lower semicontinuous on $\{\rho\geq0\}\times\Hcal$.
\end{lemma}

\begin{proof}
By L\"owner's theorem, an operator monotone $f$ on $(0,\infty)$ is automatically operator \emph{concave}. By Ando's theorem on the joint concavity of two-variable operator means~\cite{Ando1979} (see also the perspective-function formulation of Effros~\cite{Effros2009} and Hiai--Petz~\cite{HiaiPetz2012}), for such $f$ the associated Kubo--Ando mean is jointly concave: in our notation, for every fixed $z\in\Hcal$ the map $\rho\mapsto\langle z,\rhohat_{\rho,f}z\rangle$ is concave on $\{\rho\geq0\}$ (this is exactly Definition~\ref{def:kubo-ando} read as the two-variable perspective $m_f(L_\rho,R_\rho)$ of the operator concave $f$). Fix $\rho\geq0$ and $m\in\operatorname{ran}(\rhohat_{\rho,f})$; the standard variational (Legendre-type) representation of the inverse quadratic form,
\[
\tfrac12\langle m,\rhohat_{\rho,f}^{-1}m\rangle = \sup_{z\in\Hcal}\Bigl\{ \operatorname{Re}\langle m,z\rangle - \tfrac12\langle z,\rhohat_{\rho,f}z\rangle \Bigr\},
\]
exhibits $\ell_{G^{(f)}}(\rho,m)$, for $m\in\operatorname{ran}(\rhohat_{\rho,f})$, as a pointwise supremum over $z$ of maps that are linear in $m$ and, by the joint concavity just cited, convex in $\rho$ for each fixed $z$; each such map is therefore jointly convex in $(\rho,m)$, and a pointwise supremum of jointly convex functions is jointly convex. The same variational formula equals $+\infty$ exactly when $m\notin\operatorname{ran}(\rhohat_{\rho,f})$ (the supremum over $z$ in directions outside the range is unbounded), so it in fact represents $\ell_{G^{(f)}}$ on all of $\{\rho\geq0\}\times\Hcal$, proving joint convexity there. Lower semicontinuity follows since a pointwise supremum of jointly continuous functions is l.s.c.; each $z$-indexed map above is continuous in $(\rho,m)$ on $\{\rho\geq0\}\times\Hcal$ because, in finite dimension, $\Lambda_f$ is continuous on all of $[0,\infty)^2$ (it is bounded by the betweenness estimate~\eqref{eq:betweenness} and given by the same closed-form expression on $(0,\infty)^2$, extended continuously to the boundary by the defining limit $\Lambda_f(x,0)=\lim_{y\downarrow0}\Lambda_f(x,y)$), so $\rho\mapsto\rhohat_{\rho,f}$, and hence $\rho\mapsto\langle z,\rhohat_{\rho,f}z\rangle$, is continuous up to the boundary $\{\rho\geq0\}$, not merely on $\{\rho>0\}$.
\end{proof}

This closes, with a self-contained argument rather than a direct citation of~\cite{Petz1996} (whose classification theorem characterizes monotone metrics but does not by itself supply this joint-convexity statement), the gap identified in Remark~\ref{rem:existence-complexity}. The case $f=f_{\log}$ recovers exactly the joint convexity of $\ell$ used in Step~3 of Theorem~\ref{thm:existence} (there attributed to~\cite{CarlenMaas2017,CarlenMaas2020}), consistently with Lemma~\ref{lem:petz-joint-convex} being the natural generalization of that fact to the full Kubo--Ando family.

\begin{theorem}[Existence for Petz-class complexity metrics]
\label{thm:existence-petz}
Under Assumption~\ref{ass:existence}, for every \emph{range-regular} $f$ as in Definition~\ref{def:kubo-ando}, the infimum of $\Acal_{G^{(f)}}$ over $\mathcal C(\rho_0,\rho_1,T)$ is attained.
\end{theorem}

\begin{proof}
Write $m_t:=\rhohat_tv_t$ as in Step~1 of the proof of Theorem~\ref{thm:existence}, so that, by the general momentum-space cost of Remark~\ref{rem:existence-complexity}, $\ell_{G^{(f)}}(\rho,m)=\tfrac12\,m^*\rhohat_\rho^{-1/2}G_\rho^{(f)}\rhohat_\rho^{-1/2}m$. By Definition~\ref{def:petz-metric} and commutativity, $\rhohat_\rho^{-1/2}G_\rho^{(f)}\rhohat_\rho^{-1/2}=G_\rho^{(f)}\rhohat_\rho^{-1}=\rhohat_{\rho,f}^{-1}$, so
\[
\ell_{G^{(f)}}(\rho,m) = \tfrac12\,m^*\rhohat_{\rho,f}^{-1}m,
\]
consistently with Lemma~\ref{lem:petz-joint-convex}.

\emph{Coercivity.} Since $\rho\geq0$, $\tr(\rho)=1$, every eigenvalue $\lambda_j\in[0,1]$, so $\max(\lambda_j,\lambda_k)\leq1$ and~\eqref{eq:betweenness} gives $\Lambda_f(\lambda_j,\lambda_k)\leq1$ for every $j,k$, i.e.\ $\rhohat_{\rho,f}\preceq\mathrm{Id}_\Hcal$, hence $\rhohat_{\rho,f}^{-1}\succeq\mathrm{Id}_\Hcal$ on its range and
\[
\ell_{G^{(f)}}(\rho,m) \geq \tfrac12\|m\|_\Hcal^2,
\]
exactly the bound~\eqref{eq:coercivity} of Step~2, with the same constant, for every admissible $f$.

\emph{Joint convexity and lower semicontinuity} of $\ell_{G^{(f)}}$ is Lemma~\ref{lem:petz-joint-convex}.

With coercivity and joint convexity/l.s.c.\ established for $\ell_{G^{(f)}}$ with exactly the same constants and structure as $\ell$, Steps~4 and~5 of the proof of Theorem~\ref{thm:existence} apply verbatim, with $\ell$ replaced by $\ell_{G^{(f)}}$ throughout, and produce a minimizing pair $(\rho^*,m^*)$ with $\ell_{G^{(f)}}(\rho^*,m^*)<\infty$, i.e.\ $m^*_t\in\operatorname{ran}(\rhohat_{\rho^*_t,f})$ almost everywhere.

\emph{Recovery of $v^*$.} For $(\rho^*,m^*)$ to attain the infimum within $\mathcal C(\rho_0,\rho_1,T)$, its momentum must be compatible with the original mobility. The curves in $\mathcal C(\rho_0,\rho_1,T)$ are required to satisfy $\dot\rho_t+\partial^*(\rhohat_tv_t)=0$, i.e.\ $m_t=\rhohat_tv_t\in\operatorname{ran}(\rhohat_t)$, not merely $m_t\in\operatorname{ran}(\rhohat_{t,f})$. Thus one needs $\operatorname{ran}(\rhohat_{\rho,f})=\operatorname{ran}(\rhohat_\rho)$ for every $\rho\geq0$, so that the minimizing $m^*_t$ found above is automatically of the required form.

By Remark~\ref{rem:lambda-f-zero-set}, the $(j,k)$-eigenvalue $\Lambda_f(\lambda_j,\lambda_k)$ of $\rhohat_{\rho,f}$ vanishes iff at least one of $\lambda_j,\lambda_k$ is zero, matching $\Lambda_{\log}$ exactly, provided $f$ is range-regular. Without this hypothesis, mixed directions with exactly one zero eigenvalue remain in $\operatorname{ran}(\rhohat_{\rho,f})$, which is precisely the boundary mismatch illustrated by Example~\ref{ex:arithmetic-mean}. Hence, for range-regular $f$, $\operatorname{ran}(\rhohat_{\rho,f})=\operatorname{ran}(\rhohat_\rho)$ for every $\rho\geq0$ (Remark~\ref{rem:petz-boundary}). Therefore $m^*_t\in\operatorname{ran}(\rhohat_{\rho^*_t})$ a.e., and $v^*_t:=\rhohat_{\rho^*_t}^\dagger m^*_t$ (Moore--Penrose pseudoinverse) satisfies $\rhohat_{\rho^*_t}v^*_t=m^*_t$, giving $(\rho^*,v^*)\in\mathcal C(\rho_0,\rho_1,T)$ as required.
\end{proof}

\begin{remark}[Boundary behavior for non-faithful $\rho$, and why range-regularity is needed]
\label{rem:petz-boundary}
For $\rho$ with a kernel, $\Lambda_f(\lambda_j,0)=\lim_{y\downarrow0}yf(\lambda_j/y)$ may vanish, as for the range-regular logarithmic mean, or be strictly positive, as for $f=f_{\mathrm{arith}}$ in Example~\ref{ex:arithmetic-mean}. Thus $\operatorname{ran}(\rhohat_{\rho,f})$ need not equal $\operatorname{ran}(\rhohat_\rho)$ in general. Correspondingly, $G_\rho^{(f)}=\rhohat_{\rho,f}^{-1}\rhohat_\rho$ of Definition~\ref{def:petz-metric} is, strictly, an operator only on the common domain $\operatorname{ran}(\rhohat_\rho)\cap\operatorname{ran}(\rhohat_{\rho,f})$.

This boundary mismatch is not merely a matter of extending $G_\rho^{(f)}$ by convention. Working instead with the closed, everywhere-defined perspective $\ell_{G^{(f)}}(\rho,m)$ of Lemma~\ref{lem:petz-joint-convex} does not sidestep it: the direct method can in principle produce a minimizing momentum $m^*\in\operatorname{ran}(\rhohat_{\rho^*,f})\setminus\operatorname{ran}(\rhohat_{\rho^*})$ whenever this difference is nonempty. Such an $m^*$ is not of the form $\rhohat_{\rho^*}v^*$ for any velocity $v^*\in\Hcal$, which is precisely the admissibility requirement defining $\mathcal C(\rho_0,\rho_1,T)$ (Assumption~\ref{ass:existence}). This is why Theorem~\ref{thm:existence-petz} is restricted to range-regular $f$: exactly this restriction forces $\operatorname{ran}(\rhohat_{\rho,f})=\operatorname{ran}(\rhohat_\rho)$ for every $\rho$, closing the gap. For $f$ that are not range-regular, Theorem~\ref{thm:existence-petz} as stated does not apply; Example~\ref{ex:arithmetic-mean} discusses this case.
\end{remark}

\begin{example}[Arithmetic-mean complexity metric]
\label{ex:arithmetic-mean}
Take $f(t)=(1+t)/2$, operator monotone (affine functions are operator monotone) and normalized as required. Then $\Lambda_f(x,y)=(x+y)/2$, and
\[
G_\rho^{(f)}\Big|_{jk} = \frac{\Lambda(\lambda_j,\lambda_k)}{\bigl(\lambda_j+\lambda_k\bigr)/2} = \frac{2(\lambda_j-\lambda_k)}{(\lambda_j+\lambda_k)\log(\lambda_j/\lambda_k)} \qquad(\lambda_j\neq\lambda_k),
\]
equal to $1$ when $\lambda_j=\lambda_k$, and decreasing towards $0$ as the population ratio $\lambda_j/\lambda_k$ becomes increasingly unbalanced (in particular $\Lambda_f(x,0)=x/2>0$, so Remark~\ref{rem:petz-boundary} is genuinely relevant here, unlike for $f_{\log}$); a direct check shows this ratio lies in $(0,1]$. So $G_\rho^{(f)}$ is a genuinely anisotropic, bounded, $\rho$-dependent complexity metric on the velocity $v$, maximal (equal to $1$) on directions connecting nearly equal populations and decreasing toward $0$ on directions connecting very unequal populations. In the dynamically relevant momentum variable $m=\rhohat_\rho v$, the effective cost per direction is $1/\Lambda_f(\lambda_j,\lambda_k)$, so it is more precise to say that this metric penalizes \emph{velocity} along nearly-equal-population directions $(j,k)$ more heavily than along directions with very different populations — equivalently, it makes \emph{momentum} cheap precisely in those nearly-equal-population directions. Since $f_{\mathrm{arith}}$ is not range-regular ($\Lambda_f(x,0)=x/2>0$, Remark~\ref{rem:petz-boundary}), this metric falls \emph{outside} the hypothesis of Theorem~\ref{thm:existence-petz} for endpoints or intermediate states with a kernel: $\ell_{G^{(f)}}$ is still a well-defined, jointly convex, closed momentum cost (Lemma~\ref{lem:petz-joint-convex} applies to every $f$ of Definition~\ref{def:kubo-ando}, not only range-regular ones), but its minimizer need not be recoverable as $m^*=\rhohat_{\rho^*}v^*$ for the \emph{original} mobility $\rhohat_\rho$. What it defines instead is a well-posed alternative transport problem with its own, arithmetic-mean mobility $\rhohat_{\rho,f_{\mathrm{arith}}}$ — an instance, in the terminology of Section~\ref{sec:calculus-deformation}, of a different noncommutative mobility rather than a deformation of the cost alone. We do not pursue existence for this alternative problem here.
\end{example}

\begin{remark}[What remains outside the Petz-class theorem]
\label{rem:existence-residual}
Theorem~\ref{thm:existence-petz} covers every complexity metric of the form $G_\rho^{(f)}$ built from a \emph{range-regular} operator mean (Definition~\ref{def:kubo-ando}) — every $G_\rho$ diagonal in the eigenbasis of $\rho$ for which $\rhohat_{\rho,f}$ and $\rhohat_\rho$ share the same range at every $\rho$, including the logarithmic mean itself. This excludes non-range-regular means such as the arithmetic mean (Example~\ref{ex:arithmetic-mean}), for which existence for the original logarithmic mobility is not settled here. It also excludes, as a Petz-class statement, weights diagonal in a basis \emph{unrelated} to that of $\rho$ (such as a fixed basis of inner derivations $X_j$, as used throughout Sections~\ref{sec:qubit-example}--\ref{sec:ghz-example}) evaluated along a curve where this fixed basis and the eigenbasis of $\rho_t$ genuinely differ and rotate relative to one another. The latter fixed-physical-weight case is instead handled separately in Theorem~\ref{thm:existence-fixed-G}, for uniformly positive fixed $G$, by an argument that does not rely on joint convexity. Section~\ref{sec:calculus-deformation}, Remark~\ref{rem:deformation-vs-petz}, resolves a different but closely related problem obtained by deforming the continuity equation along with the cost.
\end{remark}

%
%

\subsection{Existence for fixed physical complexity weights}
\label{ssec:fixed-physical-existence}

Remark~\ref{rem:existence-residual} left open the case of a complexity metric
penalizing a fixed set of physical directions --- $G_\rho\equiv G$
independent of $\rho$, positive and bounded but diagonal in a basis
\emph{unrelated} to the (rotating) eigenbasis of $\rho_t$, as with the Pauli
weights~\eqref{eq:bell-weights}. We settle it here, in finite dimensions,
in the affirmative. The point is that the obstruction identified in
Remark~\ref{rem:existence-complexity} --- failure of \emph{joint} convexity as
the relevant structural mechanism for the momentum cost $\ell_G$ --- is visible
already in the single-qubit diagnostic below (Remark~\ref{rem:ellG-not-jointly-convex}),
but is not what the direct method actually requires: the
lower-semicontinuity step needs only convexity of $\ell_G(\rho,\cdot)$ in the
\emph{momentum} variable, which holds automatically, together with joint lower
semicontinuity of $\ell_G$ as an extended-real integrand. This is strictly
weaker than the joint convexity used in Step~3 of Theorem~\ref{thm:existence}
and in Lemma~\ref{lem:petz-joint-convex}, and it does not fail.

Throughout, $\M=M_n(\mathbb C)$, Assumption~\ref{ass:existence} is in force, and
$G\in\mathcal B(\Hcal)$ is a fixed self-adjoint operator with
\begin{equation}
\label{eq:G-fixed-bounds}
cI_\Hcal \;\leq\; G \;\leq\; CI_\Hcal, \qquad 0<c\leq C<\infty .
\end{equation}
No bimodule-compatibility (Definition~\ref{def:compatible-complexity}) and no
commutativity with $\rhohat_\rho$ are assumed; in particular $G$ may be diagonal
in any fixed orthonormal basis of derivations. The penalized action is
$\Acal_G$ of~\eqref{eq:action-G} with $G_\rho\equiv G$, subject to the
unchanged continuity equation~\eqref{eq:continuity}, and the associated
momentum cost is (Remark~\ref{rem:existence-complexity})
\begin{equation}
\label{eq:ellG-fixed}
\ell_G(\rho,m) \;=\; \tfrac12\,\bigl\langle \rhohat_\rho^{-1/2}m,\;
G\,\rhohat_\rho^{-1/2}m\bigr\rangle_\Hcal
\;=\; \tfrac12\,\bigl\|G^{1/2}\rhohat_\rho^{-1/2}m\bigr\|_\Hcal^2,
\qquad m\in\operatorname{ran}(\rhohat_\rho),
\end{equation}
extended by $\ell_G(\rho,m):=+\infty$ for $m\notin\operatorname{ran}(\rhohat_\rho)$.

\subsubsection*{The obstruction is genuine}

\begin{remark}[Numerical obstruction to joint convexity for non-scalar fixed $G$]
\label{rem:ellG-not-jointly-convex}
Let $\M=M_2(\mathbb C)$, let $\Hcal\cong M_2(\mathbb C)$ carry the standard
bimodule structure of Example~\ref{ex:finite-dim}, and let $G$ be diagonal in
the orthonormal Pauli basis $\{\tfrac{1}{\sqrt2}\sigma_\alpha\}_{\alpha\in\{0,x,y,z\}}$
with weights $(g_0,g_x,g_y,g_z)=(1,1,1,g)$, $g\neq1$.
Since $\rhohat_\rho$ is diagonal in the eigenbasis of $\rho$ with entries
$\Lambda(\lambda_j,\lambda_k)$ (Example~\ref{ex:finite-dim}) while $G$ is
diagonal in the fixed Pauli basis, the two operators do not commute for a
generic faithful $\rho$ whose eigenbasis differs from the Pauli basis, so
$\rhohat_\rho^{-1/2}G\,\rhohat_\rho^{-1/2}$ is genuinely $\rho$-dependent in a
non-scalar way.

A finite-difference Hessian computation in Hermitian Pauli coordinates provides
a simple diagnostic for the resulting loss of joint convexity. For instance,
with weights $(1,1,1,5)$ the least Hessian eigenvalue detected at interior
points $\rho\succ0$ is of order $-6.10$, and the associated midpoint test gives
\[
\tfrac12\bigl(\ell_G(\rho_0,m_0)+\ell_G(\rho_1,m_1)\bigr)
-
\ell_G\!\left(\tfrac{\rho_0+\rho_1}{2},\tfrac{m_0+m_1}{2}\right)
\approx -1.99\cdot10^{-2}.
\]
The isotropic control case $(1,1,1,1)=I_\Hcal$ gives a numerically nonnegative
Hessian, in agreement with the joint convexity used in
Theorem~\ref{thm:existence}. This diagnostic is not an input to the existence
proof below; its role is only to explain why the proof should not rely on joint
convexity for arbitrary fixed physical weights.
\end{remark}

Remark~\ref{rem:ellG-not-jointly-convex} indicates that no verbatim extension
of Step~3 of Theorem~\ref{thm:existence} or of Lemma~\ref{lem:petz-joint-convex}
should be expected for a fixed, non-scalar $G$: the obstruction predicted in
Remark~\ref{rem:existence-complexity} occurs already on a single qubit. What
follows shows that existence nonetheless holds, by replacing joint convexity
with the weaker pair of properties that the direct method genuinely uses.

\subsubsection*{Structural properties of $\ell_G$}

\begin{lemma}[Coercivity, partial convexity, and lower semicontinuity]
\label{lem:ellG-structural}
Let $G$ satisfy~\eqref{eq:G-fixed-bounds}. Then:
\begin{enumerate}
\item[\textup{(i)}] \emph{(Domain and coercivity.)} $\ell_G(\rho,m)<\infty$ if and
only if $m\in\operatorname{ran}(\rhohat_\rho)$ --- the same finiteness domain as
the unweighted cost $\ell$ of Definition~\ref{def:momentum} --- and for all $\rho\geq0$
with $\tr(\rho)=1$,
\begin{equation}
\label{eq:ellG-coercive}
\ell_G(\rho,m)\;\geq\;\tfrac{c}{2}\,\bigl\langle m,\rhohat_\rho^{-1}m\bigr\rangle_\Hcal
\;\geq\;\tfrac{c}{2}\,\|m\|_\Hcal^2 .
\end{equation}
\item[\textup{(ii)}] \emph{(Convexity in the momentum.)} For every fixed
$\rho\geq0$, the map $m\mapsto\ell_G(\rho,m)$ is convex --- indeed a
nonnegative quadratic form --- on $\Hcal$.
\item[\textup{(iii)}] \emph{(Joint lower semicontinuity.)} The map
$(\rho,m)\mapsto\ell_G(\rho,m)$ is lower semicontinuous on
$\{\rho\geq0\}\times\Hcal$ with values in $[0,+\infty]$; equivalently, it is a
nonnegative normal integrand.
\end{enumerate}
\end{lemma}

\begin{proof}
\emph{(i)} By~\eqref{eq:G-fixed-bounds}, $\tfrac{c}{2}\|\rhohat_\rho^{-1/2}m\|^2
\leq\ell_G(\rho,m)\leq\tfrac{C}{2}\|\rhohat_\rho^{-1/2}m\|^2$, and
$\|\rhohat_\rho^{-1/2}m\|^2=\langle m,\rhohat_\rho^{-1}m\rangle=2\,\ell(\rho,m)$
is finite exactly on $\operatorname{ran}(\rhohat_\rho)$; this gives the domain
identification and the first inequality in~\eqref{eq:ellG-coercive}. The second
is~\eqref{eq:coercivity}: $\tr(\rho)=1$ forces every eigenvalue $\lambda_j\in[0,1]$,
so $\Lambda(\lambda_j,\lambda_k)\leq\max(\lambda_j,\lambda_k)\leq1$, i.e.\ $\rhohat_\rho\preceq I_\Hcal$ and $\rhohat_\rho^{-1}\succeq I_\Hcal$ on its range.

\emph{(ii)} For fixed $\rho$, $S(\rho):=\rhohat_\rho^{-1/2}G\,\rhohat_\rho^{-1/2}$
is a fixed positive semidefinite operator on $\operatorname{ran}(\rhohat_\rho)$,
and $\ell_G(\rho,m)=\tfrac12\langle m,S(\rho)m\rangle$ is the associated
quadratic form, convex in $m$; on the complement of $\operatorname{ran}(\rhohat_\rho)$
it is the convex function $+\infty$. (This is the only convexity used below, and
it holds with no commutativity hypothesis on $G$ and $\rhohat_\rho$,
unlike the joint convexity of Lemma~\ref{lem:petz-joint-convex}.)

\emph{(iii)} Write $B:=G^{1/2}$, a fixed bounded, boundedly invertible positive
operator. The Legendre--Fenchel dual of the quadratic form $m\mapsto\ell_G(\rho,m)$
gives the variational representation
\begin{equation}
\label{eq:ellG-sup}
\ell_G(\rho,m)\;=\;\sup_{y\in\Hcal}\;\Phi_y(\rho,m),
\qquad
\Phi_y(\rho,m):=\operatorname{Re}\langle m,y\rangle_\Hcal
-\tfrac12\bigl\|B^{-1}\rhohat_\rho^{1/2}\,y\bigr\|_\Hcal^2 ,
\end{equation}
valid for all $(\rho,m)\in\{\rho\geq0\}\times\Hcal$, including
$m\notin\operatorname{ran}(\rhohat_\rho)$ (where the supremum is $+\infty$, taken
along $y$ with a component in $\ker\rhohat_\rho^{1/2}$ on which
$\langle m,y\rangle\neq0$). Indeed, on $\operatorname{ran}(\rhohat_\rho)$,
\eqref{eq:ellG-sup} is the standard identity
$\tfrac12\langle m,A^{-1}m\rangle=\sup_y\{\operatorname{Re}\langle m,y\rangle-\tfrac12\langle y,Ay\rangle\}$
with $A=S(\rho)^{-1}=\rhohat_\rho^{1/2}G^{-1}\rhohat_\rho^{1/2}=(B^{-1}\rhohat_\rho^{1/2})^*(B^{-1}\rhohat_\rho^{1/2})$.

Now fix $y$. The map $\rho\mapsto\rhohat_\rho=\int_0^1 L_{\rho^s}R_{\rho^{1-s}}\,\dd s$
of~\eqref{eq:rhohat-def} is continuous on $\{\rho\geq0\}$ \emph{up to the
boundary} (each $\rho\mapsto\rho^s$, $s\in[0,1]$, is continuous on the positive
cone, and $L,R$ are continuous), hence so is $\rho\mapsto\rhohat_\rho^{1/2}$ (the
operator square root is continuous on the positive cone) and, $B$ being fixed,
so is $\rho\mapsto\|B^{-1}\rhohat_\rho^{1/2}y\|^2$. Together with the continuous
linear term $\operatorname{Re}\langle m,y\rangle$, this makes each $\Phi_y$
jointly continuous on $\{\rho\geq0\}\times\Hcal$. A pointwise supremum of
jointly continuous functions is lower semicontinuous, so~\eqref{eq:ellG-sup}
exhibits $\ell_G$ as jointly l.s.c. Nonnegativity is clear from~\eqref{eq:ellG-fixed}.
Finally, $\ell_G$ has no explicit $t$-dependence, so joint lower semicontinuity
is exactly the statement that $\ell_G$ is a nonnegative normal integrand.
\end{proof}

\begin{remark}[Why partial convexity suffices]
\label{rem:partial-convexity-suffices}
The representation~\eqref{eq:ellG-sup} does not make $\ell_G$ jointly
convex: each $\Phi_y$ is concave in $\rho$ through $-\tfrac12\|B^{-1}\rhohat_\rho^{1/2}y\|^2$
(as $\rho\mapsto\rhohat_\rho$ is operator concave) but affine in $m$, so it is
neither jointly convex nor jointly concave, in accordance with the diagnostic in
Remark~\ref{rem:ellG-not-jointly-convex}. What survives is precisely
Lemma~\ref{lem:ellG-structural}(ii)--(iii): convexity in the momentum alone,
plus joint lower semicontinuity. This is exactly the hypothesis of the classical
lower-semicontinuity theorem for integral functionals under weak convergence of
the convex variable (Ioffe~\cite{Ioffe1977}; see also
Buttazzo~\cite{Buttazzo1989} and Fonseca--Leoni~\cite{FonsecaLeoni2007}), which
is all that Steps~4--5 of Theorem~\ref{thm:existence}
invoke.
\end{remark}

\subsubsection*{Existence}

\begin{theorem}[Existence for fixed physical complexity weights]
\label{thm:existence-fixed-G}
Under Assumption~\ref{ass:existence}, for every fixed self-adjoint
$G\in\mathcal B(\Hcal)$ satisfying~\eqref{eq:G-fixed-bounds} --- with no
commutativity or bimodule-compatibility hypothesis, in particular for any $G$
diagonal in a fixed basis of physical derivations --- the infimum of $\Acal_G$
over $\mathcal C(\rho_0,\rho_1,T)$ is attained. There exists
$(\rho^*,v^*)\in\mathcal C(\rho_0,\rho_1,T)$ with
$\Acal_G(\rho^*,v^*)=\inf_{\mathcal C(\rho_0,\rho_1,T)}\Acal_G$.
\end{theorem}

\begin{proof}
Work in the momentum variable $m_t:=\rhohat_{\rho_t}v_t$ as in
Definition~\ref{def:momentum}, so that the continuity
equation~\eqref{eq:continuity} becomes the linear relation
$\dot\rho_t+\partial^*m_t=0$ and $\Acal_G(\rho,v)=\int_0^T\ell_G(\rho_t,m_t)\,\dd t$
with $\ell_G$ as in~\eqref{eq:ellG-fixed}.

\emph{Compactness.} Let $(\rho^k,m^k)$ be a minimizing sequence over the
admissible set (nonempty by Assumption~\ref{ass:existence}). By the coercivity
bound~\eqref{eq:ellG-coercive},
\[
\frac{c}{2}\|m^k\|_{L^2([0,T];\Hcal)}^2
\leq \Acal_G(\rho^k,v^k)
\leq \sup_k\Acal_G(\rho^k,v^k)<\infty .
\]
Thus, along a subsequence, $m^k\rightharpoonup m^*$ weakly in
$L^2([0,T];\Hcal)$. Since
$\dot\rho^k=-\partial^*m^k$ with $\partial^*$ a fixed bounded operator,
$\{\rho^k\}$ is bounded in $H^1([0,T];\M)$, which embeds compactly into
$C([0,T];\M)$; passing to a further subsequence, $\rho^k\to\rho^*$ uniformly.
The constraints $\rho_t\geq0$, $\tr(\rho_t)=1$ and the boundary conditions are
closed under uniform convergence, and $\dot\rho^*=-\partial^*m^*$ follows by
passing to the limit in the distributional identity. Hence $(\rho^*,m^*)$ is
admissible. (This is Step~4 of Theorem~\ref{thm:existence} verbatim; only the
coercivity constant differs, $\tfrac12\mapsto\tfrac{c}{2}$.)

\emph{Lower semicontinuity.} By Lemma~\ref{lem:ellG-structural}, $\ell_G\geq0$ is
a normal integrand, jointly l.s.c. and convex in its second argument $m$ for each
fixed $\rho$. The minimizing sequence has $\rho^k\to\rho^*$ uniformly (hence in
measure) and $m^k\rightharpoonup m^*$ weakly in $L^2\subset L^1$, with $\{m^k\}$
bounded in $L^2$ over a finite interval and therefore uniformly integrable. Ioffe's
lower-semicontinuity theorem for integral functionals convex in the weakly
converging variable~\cite{Ioffe1977} then gives
\[
\int_0^T\ell_G(\rho^*_t,m^*_t)\,\dd t
\;\leq\;\liminf_{k}\int_0^T\ell_G(\rho^k_t,m^k_t)\,\dd t
\;=\;\inf_{\mathcal C(\rho_0,\rho_1,T)}\Acal_G .
\]

\emph{Recovery of the velocity.} The left-hand side is finite, so
$\ell_G(\rho^*_t,m^*_t)<\infty$ for a.e.\ $t$; by Lemma~\ref{lem:ellG-structural}(i)
this is equivalent to $m^*_t\in\operatorname{ran}(\rhohat_{\rho^*_t})$ a.e. --- the
same range as the unweighted mobility, with no range-regularity caveat
(Remark~\ref{rem:petz-boundary}), because here the mobility is the logarithmic
mean $\rhohat$ itself and $G$ is merely a bounded, boundedly invertible factor.
Setting $v^*_t:=\rhohat_{\rho^*_t}^{\dagger}m^*_t$ (Moore--Penrose pseudoinverse)
gives $\rhohat_{\rho^*_t}v^*_t=m^*_t$, so $(\rho^*,v^*)\in\mathcal C(\rho_0,\rho_1,T)$
and $\Acal_G(\rho^*,v^*)=\int_0^T\ell_G(\rho^*_t,m^*_t)\,\dd t
\leq\inf_{\mathcal C}\Acal_G$. Hence the infimum is attained.
\end{proof}

\begin{remark}[Scope: what is and is not obtained]
\label{rem:fixed-G-scope}
Theorem~\ref{thm:existence-fixed-G} closes the residual case flagged in
Remark~\ref{rem:existence-residual}: existence of minimizers for the
cost-only penalized action $\Acal_G$ of~\eqref{eq:action-G}, under the
unchanged continuity equation~\eqref{eq:continuity}, for any fixed,
uniformly positive $G$ --- including Pauli-weight metrics diagonal in a basis
unrelated to the eigenbasis of $\rho_t$. Three points delimit it. First, it is a
finite-dimensional statement ($\M=M_n$), like Theorem~\ref{thm:existence}; the
semifinite case is not addressed. Second, it uses~\eqref{eq:G-fixed-bounds}
(uniform positivity $c>0$) essentially, both for coercivity and to identify the
finiteness domain with $\operatorname{ran}(\rhohat)$; a merely nonnegative fixed
$G$ with nontrivial kernel would require a separate range analysis. Third, and
unlike Theorem~\ref{thm:existence} and Lemma~\ref{lem:petz-joint-convex}, the
argument yields \emph{no} joint convexity of $\ell_G$ --- the numerical
diagnostic of Remark~\ref{rem:ellG-not-jointly-convex} shows why this is the
wrong property to rely on for non-scalar fixed weights --- so it does not by
itself deliver geodesic convexity of the entropy or uniqueness of minimizers for
$\Acal_G$; the fixed-physical-weight geometry may be genuinely non-convex, which
is consistent with the anisotropic single-qubit and entangling examples of
Sections~\ref{sec:qubit-example}--\ref{sec:ghz-example}. What is obtained is
exactly existence, via partial convexity (in the momentum) plus lower
semicontinuity, superseding the stronger-than-necessary joint-convexity
requirement stated in Remark~\ref{rem:existence-complexity}.
\end{remark}

\begin{remark}[Relation to the calculus deformation of Section~\ref{sec:calculus-deformation}]
\label{rem:fixed-G-vs-deformation}
Theorem~\ref{thm:complexity-deformation} handles a fixed, bimodule-compatible $G$
by absorbing it into the derivation, $\partial_G=G^{1/2}\partial$, thereby
deforming the continuity equation to $\dot\rho_t+\partial^*(G\rhohat_tv_t)=0$
and reducing to unweighted transport for $\partial_G$. Theorem~\ref{thm:existence-fixed-G}
is complementary: it keeps the original mobility $\rhohat_tv_t$ and the original
continuity equation~\eqref{eq:continuity} fixed, reweighting only the cost, and
requires \emph{no} bimodule compatibility --- $G$ need not commute with $L_a,R_a$,
so genuinely physical Pauli weights, excluded from
Definition~\ref{def:compatible-complexity}, are covered. The two results thus
resolve the two inequivalent variational problems of
Remark~\ref{rem:cost-versus-calculus-deformation} --- cost-only deformation
versus joint deformation of cost and mobility --- for the same class of fixed
weights, by different mechanisms. Whether their minimizers coincide in any case
of interest remains the open question recorded there.
\end{remark}

\section{Geometry of unitary orbits}
\label{sec:orbits}

The mixed-state transport theory of Sections~\ref{sec:transport}--\ref{sec:existence} is complemented by a logically distinct sector, governed directly by the Liouville--von Neumann equation on a single unitary orbit (Section~\ref{sec:hamiltonian}). We show that the resulting orbit distance is exactly a quotient metric induced by a right-invariant complexity geometry on the ambient unitary group (Section~\ref{sec:bridge}), which is Theorem~\ref{thm:orbit-geometry} of the Introduction, and we illustrate the whole construction with a complete worked example, the anisotropic single qubit (Section~\ref{sec:qubit-example}), including its interpretation as a torque-free rigid-body rotation (Section~\ref{ssec:bridge-euler-arnold}).

\subsection{The Hamiltonian (unitary-orbit) sector}
\label{sec:hamiltonian}

When transport is restricted to inner automorphisms, the generator $\Lcal_t$ is an inner derivation
\begin{equation}
\label{eq:hamiltonian-derivation}
\delta_{H_t}(a) = i[H_t,a], \qquad H_t = H_t^* \text{ affiliated with } \M.
\end{equation}
Equation~\eqref{eq:dual-state} on the states reads
\[
\frac{\dd}{\dd t}\,\omega_t(a) = \omega_t\bigl(i[H_t,a]\bigr), \qquad a \in \dom(\delta_{H_t}),
\]
and, under Assumption~\ref{ass:trace}, is equivalent to
\begin{equation}
\label{eq:liouville}
\dot\rho_t = -i[H_t,\rho_t],
\end{equation}
the \emph{Liouville--von Neumann equation}.

\begin{remark}
\label{rem:spectrum-preserved}
Equation~\eqref{eq:liouville} generates a unitary flow, $\rho_t = U_t \rho_0 U_t^*$, and therefore preserves the spectrum of $\rho_0$: the eigenvalues of $\rho_t$ do not depend on $t$. Consequently, \eqref{eq:liouville} only connects states lying on the same unitary orbit, and cannot in general transport an arbitrary state $\rho_0$ to another state $\rho_1$ with different spectrum. This limitation is the structural motivation for introducing, in the next section, a differential calculus that allows the spectrum of the density to vary.

We treat these as two related but logically distinct sectors of the theory. The Hamiltonian, or unitary-orbit, sector of this section is governed by~\eqref{eq:liouville} directly at the level of the dual dynamics on observables, with no $\rhohat$ involved. The mixed-state sector of Sections~\ref{sec:continuity}--\ref{sec:complexity} is governed by the full continuity equation~\eqref{eq:continuity} with $\rhohat$. Sections~\ref{sec:qubit-example}--\ref{sec:ghz-example} work entirely within the unitary-orbit sector: the endpoints have the same spectrum and no $\rhohat$ is invoked, as noted explicitly in Section~\ref{ssec:reduction}. Section~\ref{sec:bridge} makes precise, via Theorem~\ref{thm:orbit-geometry}, the sense in which the orbit distance of that sector coincides with the geometric complexity of unitary dilations. This is a statement about the unitary-orbit sector specifically, not a claim that it is derived from, or identical to, the general mixed-state functional $\Acal_G$ of~\eqref{eq:action-G} restricted to pure states.
\end{remark}

\subsection{Right-invariant complexity geometry and the orbit-distance identity}
\label{sec:bridge}

We now make precise the connection, anticipated in Section~\ref{ssec:qubit-interpretation}, between the unitary-orbit sector of the transport problem — Sections~\ref{sec:qubit-example}--\ref{sec:ghz-example} — and the geometric complexity of unitary dilations developed by the authors in~\cite{AcevedoFalco2026}. Throughout this section $\M=M_N(\mathbb C)$, the differential calculus is generated by a fixed family of inner derivations $\partial_j a=[X_j,a]$ with Hermitian $X_j$ spanning a subalgebra $\mathfrak g:=\operatorname{span}_{\mathbb R}\{X_j\}_{j=1}^m\subset\mathfrak{su}(N)$ (Example~\ref{ex:internal-derivations}), and the complexity metric $G_\rho\equiv G=\operatorname{diag}(g_1,\dots,g_m)$ is constant in this basis — exactly the setting of Sections~\ref{sec:qubit-example}--\ref{sec:ghz-example}, with $\mathfrak g=\mathfrak{su}(2)$ in every worked example.

\subsubsection{Recollection: right-invariant \texorpdfstring{$\hat\Omega$}{Omega}-weighted geometry}
\label{ssec:bridge-recall}

Let $H$ be a compact connected Lie group with Lie algebra $\mathfrak h$, and let $\hat\Omega$ be a Hermitian positive-definite form on $\mathfrak h$. The construction of~\cite{AcevedoFalco2026} — stated there for $H=SU(N)$, but using only that $H$ is a compact Lie group — associates with $\hat\Omega$ an inner product $\langle\cdot,\cdot\rangle_{\hat\Omega}$ on $\mathfrak h$ (normalized there by $1/\operatorname{Tr}\hat\Omega$ to remove the scale ambiguity $\hat\Omega\mapsto c\hat\Omega$; since $G$ below is a fixed, given operator, this normalization plays no role and we drop it), a right-invariant Riemannian metric $g^{(\hat\Omega)}_{\hat U}(\hat A\hat U,\hat B\hat U):=\langle\hat A,\hat B\rangle_{\hat\Omega}$ on $H$, and the \emph{unitary geometric complexity}
\begin{align}
\label{eq:bridge-Gomega}
G_{\hat\Omega}(\hat U) &:= D_{\hat\Omega}(\hat I,\hat U), \\
D_{\hat\Omega}(\hat U,\hat V) &:=\inf\Bigl\{\textstyle\int_0^1\sqrt{\bigl\langle\dot\gamma(s)\gamma(s)^{-1},\dot\gamma(s)\gamma(s)^{-1}\bigr\rangle_{\hat\Omega}}\,\dd s \;:\; \gamma(0)=\hat U,\ \gamma(1)=\hat V \Bigr\}, \notag
\end{align}
equivalently, in control form, $G_{\hat\Omega}(\hat U)=\inf_{H(\cdot)}\bigl\{\int_0^1\sqrt{\langle H(s),H(s)\rangle_{\hat\Omega}}\,\dd s : \dot\gamma=-iH\gamma,\ \gamma(0)=\hat I,\ \gamma(1)=\hat U\bigr\}$. Since right-invariant metrics on Lie groups are geodesically complete and $H$ is compact, $(H,g^{(\hat\Omega)})$ is geodesically complete and $D_{\hat\Omega}$ is attained by a minimizing geodesic (Hopf--Rinow), cf.~\cite[\S3.4]{AcevedoFalco2026}.

We identify $\hat\Omega\equiv G$ via $\hat E_j:=-iX_j$, so that $\langle\hat A,\hat A\rangle_{\hat\Omega}=\mathbf h^{\mathsf T}G\mathbf h$ for $\hat A=-iH(t)$, $H(t)=\sum_j h_j(t)X_j$; under this identification $\eqref{eq:bridge-Gomega}$ literally reproduces the control problem of~\eqref{eq:qubit-action}--\eqref{eq:qubit-length}.

Let $H_{\mathfrak g}$ denote the connected subgroup of $SU(N)$ generated by $\mathfrak g$. In every example of Sections~\ref{sec:qubit-example}--\ref{sec:ghz-example}, $\mathfrak g\cong\mathfrak{su}(2)$ is semisimple, and a connected subgroup of a compact Lie group generated by a semisimple subalgebra is automatically closed; hence $H_{\mathfrak g}$ is itself a compact Lie group with Lie algebra $\mathfrak g$, and the recollection above applies to $H=H_{\mathfrak g}$ directly.

\subsubsection{The orbit-distance identity}
\label{ssec:bridge-orbit}

\mainthmletter{D}
\begin{mainthm}[Geometry of unitary orbits]
\label{prop:bridge-orbit}
\label{cor:bridge-existence}
\label{thm:orbit-geometry}
Let $\mathfrak g$, $G$, $H:=H_{\mathfrak g}$ be as above, and let $\rho_0,\rho_1$ be pure states in the same $H$-orbit under conjugation. Set
\[
K_{\rho_0} := \{U\in H : U\rho_0 U^*=\rho_0\}, \qquad \mathcal F(\rho_0,\rho_1) := \{U\in H : U\rho_0 U^*=\rho_1\},
\]
a right coset of the closed subgroup $K_{\rho_0}$ in $H$. Then:
\begin{enumerate}[label=\textnormal{(\roman*)}]
	\item \textnormal{(Orbit-distance identity.)} $d_G(\rho_0,\rho_1) = \inf_{U\in\mathcal F(\rho_0,\rho_1)} G_{\hat\Omega}(U)$, $\hat\Omega\equiv G$, i.e.\ the unitary-orbit transport distance equals the distance, in the right-invariant complexity geometry $(H,g^{(\hat\Omega)})$, from the identity to the coset $\mathcal F(\rho_0,\rho_1)$ of the stabilizer of $\rho_0$.
	\item \textnormal{(Existence of minimizers.)} The infimum in (i) is attained at some $U^*\in\mathcal F(\rho_0,\rho_1)$, and the trajectory $\rho_t:=\gamma^*(t)\rho_0\gamma^*(t)^*$, where $\gamma^*$ is a minimizing $g^{(\hat\Omega)}$-geodesic in $H$ from $\hat I$ to $U^*$, realizes $d_G(\rho_0,\rho_1)$.
\end{enumerate}
\end{mainthm}

\begin{proof}
\emph{(i).} By definition~\eqref{eq:qubit-length} (extended verbatim from $M_2(\mathbb C)$ to the present $\mathfrak g$-generated sector, as in Sections~\ref{sec:bell-example}--\ref{sec:ghz-example}), $d_G(\rho_0,\rho_1)$ is the infimum of $\int_0^T\sqrt{\mathbf h(t)^{\mathsf T}G\mathbf h(t)}\,\dd t$ over $T>0$ and controls $H(t)=\sum_jh_j(t)X_j$ such that the solution of $\dot U=-iHU$, $U(0)=\hat I$, satisfies $U(T)\rho_0U(T)^*=\rho_1$, i.e.\ $U(T)\in\mathcal F(\rho_0,\rho_1)$. For fixed such a control, the integral is exactly the $g^{(\hat\Omega)}$-length $L_{\hat\Omega}(U(\cdot))$ of the path $U(\cdot)$ in $H$. Hence
\begin{align*}
d_G(\rho_0,\rho_1) &= \inf\bigl\{L_{\hat\Omega}(\gamma) : \gamma(0)=\hat I,\ \gamma(T)\in\mathcal F(\rho_0,\rho_1)\text{ for some }T>0\bigr\} \\
&= \inf_{U\in\mathcal F(\rho_0,\rho_1)}\ \inf\bigl\{L_{\hat\Omega}(\gamma):\gamma(0)=\hat I,\gamma(T)=U\bigr\},
\end{align*}
and the inner infimum is $D_{\hat\Omega}(\hat I,U)=G_{\hat\Omega}(U)$ by definition~\eqref{eq:bridge-Gomega}, which is part~(i).

\emph{(ii).} $K_{\rho_0}$ is the stabilizer of $\rho_0$ under the (continuous) conjugation action of $H$, hence a closed subgroup of the compact group $H$, hence compact; so is the coset $\mathcal F(\rho_0,\rho_1)$. The map $U\mapsto G_{\hat\Omega}(U)=D_{\hat\Omega}(\hat I,U)$ is continuous, hence attains its infimum on the compact set $\mathcal F(\rho_0,\rho_1)$ at some $U^*$. By Hopf--Rinow on $(H,g^{(\hat\Omega)})$ (Section~\ref{ssec:bridge-recall}), $D_{\hat\Omega}(\hat I,U^*)$ is realized by a minimizing geodesic $\gamma^*$.
\end{proof}

Theorem~\ref{thm:orbit-geometry} recovers, from a single general compactness argument, the existence of the optimal trajectories to be constructed by hand below (Proposition~\ref{prop:qubit-optimal} in Section~\ref{sec:qubit-example}, and Proposition~\ref{prop:bell-clairaut} and Theorem~\ref{thm:ghz-generic} in Section~\ref{sec:bell-example}, all instances of $\mathfrak g\cong\mathfrak{su}(2)$), and extends it to every pair of pure states in the same $H$-orbit and every constant $G>0$ in this sector, not only the axially symmetric or fixed-axis families worked out explicitly there.

\subsection{A complete example: the single qubit}
\label{sec:qubit-example}

We now specialize the framework of Sections~\ref{sec:continuity}--\ref{sec:complexity} to $\M = M_2(\mathbb C)$, equipped with the trace $\tr = \operatorname{Tr}$ (so that $\tr(\rho)=1$ for every density $\rho$, consistently with Assumption~\ref{ass:trace} and with the convention used in the existence theorems of Section~\ref{sec:existence}), and compute explicitly the complexity-penalized transport between two orthogonal pure states. This example is fully self-contained and illustrates, in closed form, the role of the penalization operator $G_\rho$ introduced in Section~\ref{sec:complexity}.

\subsubsection{Bloch parametrization}
\label{ssec:bloch}

Let
\[
\Dcal_2 = \bigl\{ \rho \in \mathbb C^{2\times2} : \rho = \rho^\dagger,\; \rho \geq 0,\; \operatorname{Tr}(\rho)=1 \bigr\}
\]
be the state space of a single qubit. Every $\rho \in \Dcal_2$ can be written in terms of its Bloch vector as
\begin{equation}
\label{eq:bloch-param}
\rho = \frac{1}{2}\bigl( I + \mathbf r \cdot \boldsymbol\sigma \bigr), \qquad \mathbf r = (r_x,r_y,r_z) \in \mathbb R^3, \qquad \|\mathbf r\| \leq 1,
\end{equation}
where $\boldsymbol\sigma = (\sigma_x,\sigma_y,\sigma_z)$ is the vector of Pauli matrices. Pure states correspond exactly to $\|\mathbf r\|=1$, i.e.\ to the Bloch sphere $S^2$.

\subsubsection{Reduction to the inner-derivation calculus of Example~\ref{ex:internal-derivations}}
\label{ssec:reduction}

Fix a time-dependent control Hamiltonian
\begin{equation}
\label{eq:qubit-hamiltonian}
H(t) = \frac{1}{2}\, \mathbf h(t)\cdot\boldsymbol\sigma, \qquad \mathbf h(t) = (h_x(t),h_y(t),h_z(t)) \in \mathbb R^3.
\end{equation}
This is precisely the finite-dimensional instance of the inner-derivation calculus $\partial_j a = [X_j,a]$ of Example~\ref{ex:internal-derivations}, with generators $X_j = \sigma_j/2$, $j\in\{x,y,z\}$, and bimodule $\Hcal \cong \mathfrak{su}(2) \cong \mathbb R^3$ parametrized by $\mathbf h$.

If $\rho_0,\rho_1$ are pure states, they share the same spectrum $\{1,0\}$ and therefore lie on a single unitary orbit; by Remark~\ref{rem:spectrum-preserved}, transport between them can be realized entirely by the Liouville--von Neumann flow~\eqref{eq:liouville}, without invoking the full noncommutative-multiplication machinery $\rhohat$ of Section~\ref{sec:rhohat}. In this orbit-restricted regime, the penalization operator $G_\rho$ of Section~\ref{sec:complexity} reduces to a (possibly constant) positive-definite quadratic form directly on the control vector $\mathbf h$, rather than on the full bimodule $\Hcal$ at every base point $\rho$.

Substituting~\eqref{eq:bloch-param}--\eqref{eq:qubit-hamiltonian} into the Liouville--von Neumann equation~\eqref{eq:liouville} and using $[\sigma_a,\sigma_b] = 2i\varepsilon_{abc}\sigma_c$ gives the \emph{Bloch equation}
\begin{equation}
\label{eq:bloch-eq}
\dot{\mathbf r}(t) = \mathbf h(t) \times \mathbf r(t),
\end{equation}
which plays, in this finite-dimensional model, the role of the continuity equation~\eqref{eq:continuity} restricted to the unitary orbit $\|\mathbf r\|=1$ (note that~\eqref{eq:bloch-eq} automatically preserves $\|\mathbf r(t)\|$, consistent with Remark~\ref{rem:spectrum-preserved}).

\subsubsection{Anisotropic complexity metric and cost functional}
\label{ssec:aniso}

Let $\alpha_x,\alpha_y,\alpha_z>0$ be complexity weights assigned to the three Pauli directions, and set
\[
G = \operatorname{diag}(\alpha_x,\alpha_y,\alpha_z).
\]
As emphasized in Remark~\ref{rem:spectrum-preserved}, the unitary-orbit sector is governed directly by the dual dynamics $\dot\rho=-i[H,\rho]$, with dynamical variable $H$ (equivalently $\mathbf h$) rather than the bimodule velocity $v$, and does not invoke $\rhohat$. Motivated by — but logically independent of — the complexity metric $G_\rho$ of Section~\ref{sec:complexity}, we equip this orbit-control problem with its own quadratic cost $\Acal_T\equiv\Acal_T^{\mathrm{orb}}$ on the control algebra $\mathfrak{su}(2)\cong\mathbb R^3$ (the superscript is dropped from here on, since only this orbit-sector action is used in Sections~\ref{sec:qubit-example}--\ref{sec:bridge}):
\begin{equation}
\label{eq:qubit-action}
\Acal_T[\mathbf h] = \frac{1}{2}\int_0^T \mathbf h(t)^{\mathsf T} G\, \mathbf h(t)\,\dd t
= \frac{1}{2}\int_0^T \bigl( \alpha_x h_x(t)^2 + \alpha_y h_y(t)^2 + \alpha_z h_z(t)^2 \bigr)\,\dd t,
\end{equation}
subject to $\dot{\mathbf r}=\mathbf h\times\mathbf r$, i.e.\ $\Acal_T(\rho,H):=\Acal_T[\mathbf h]$ with $\dot\rho_t=-i[H_t,\rho_t]$ — the orbit-sector analogue of the penalized transport action~\eqref{eq:action-G}, not a specialization of it: the two functionals have different domains ($\mathbf h\in\mathbb R^3$ here, versus $v\in\Hcal$ with $\rhohat$ there) and are related only through the bridge of Section~\ref{sec:bridge}, which identifies the induced \emph{distances} on the common orbit, not the actions themselves. The associated complexity \emph{length} functional is
\begin{equation}
\label{eq:qubit-length}
d_G(\rho_0,\rho_1) = \inf_{\mathbf r,\mathbf h} \int_0^T \sqrt{\alpha_x h_x(t)^2+\alpha_y h_y(t)^2+\alpha_z h_z(t)^2}\,\dd t,
\end{equation}
the infimum taken over trajectories satisfying~\eqref{eq:bloch-eq} with $\mathbf r(0)=\mathbf r_0$, $\mathbf r(T)=\mathbf r_1$, $\|\mathbf r(t)\|=1$ — an infimum over Hamiltonian controls on a fixed unitary orbit, and not a restriction of the mixed-state distance $W_\partial$ of Definition~\ref{def:Wdist} to pure states. As in the classical Benamou--Brenier picture, $d_G$ does not depend on the time-parametrization of the curve, while for a constant-speed parametrization the minimal energy and the length are related, by the Cauchy--Schwarz inequality applied to the reparametrization-invariant length integral, as
\begin{equation}
\label{eq:energy-length}
\min_{\mathbf h} \Acal_T[\mathbf h] = \frac{d_G(\rho_0,\rho_1)^2}{2T}.
\end{equation}

\subsubsection{Transport between orthogonal pure states}
\label{ssec:orthogonal}

Consider $\rho_0 = |0\rangle\langle 0| = \tfrac12(I+\sigma_z)$ and $\rho_1 = |1\rangle\langle1| = \tfrac12(I-\sigma_z)$, with Bloch vectors $\mathbf r_0=(0,0,1)$, $\mathbf r_1=(0,0,-1)$.

\begin{proposition}
\label{prop:qubit-optimal}
Suppose $G=\operatorname{diag}(\alpha_x,\alpha_y,\alpha_z)$ with $\alpha_y \leq \alpha_x,\alpha_z$. Within the family of uniform-speed rotations about a fixed axis $\mathbf n(\varphi) = (\cos\varphi,\sin\varphi,0)$ perpendicular to the $z$-axis, the constant control
\begin{equation}
\label{eq:qubit-optimal-control}
\mathbf h_{\mathrm{opt}}(t) = \Bigl(0,\frac{\pi}{T},0\Bigr), \qquad 0\leq t\leq T,
\end{equation}
i.e.\ $H_{\mathrm{opt}} = \dfrac{\pi}{2T}\sigma_y$, minimizes the action~\eqref{eq:qubit-action} and transports $\rho_0$ to $\rho_1$ in time $T$, with
\begin{equation}
\label{eq:qubit-optimal-value}
\Acal_T^{\min} = \frac{\alpha_y \pi^2}{2T}, \qquad d_G(\rho_0,\rho_1) = \pi\sqrt{\alpha_y}.
\end{equation}
\end{proposition}

\begin{proof}
A unit vector perpendicular to the $z$-axis has the form $\mathbf n(\varphi) = (\cos\varphi,\sin\varphi,0)$. A rotation by angle $\pi$ about $\mathbf n(\varphi)$ carries $\mathbf r_0$ to $\mathbf r_1$; the corresponding uniform control is
\[
\mathbf h_\varphi(t) = \frac{\pi}{T}\bigl(\cos\varphi,\sin\varphi,0\bigr), \qquad 0 \leq t \leq T.
\]
Writing $\omega=\pi/T$ and using Rodrigues' rotation formula with $\mathbf n(\varphi)\cdot\mathbf r_0=0$, the curve
\[
\mathbf r_\varphi(t) = \bigl(\sin(\omega t)\sin\varphi,\, -\sin(\omega t)\cos\varphi,\, \cos(\omega t)\bigr)
\]
satisfies $\dot{\mathbf r}_\varphi = \mathbf h_\varphi \times \mathbf r_\varphi$ (rotation at constant angular velocity $\mathbf h_\varphi = \omega\,\mathbf n(\varphi)$ about a fixed axis), together with $\mathbf r_\varphi(0)=\mathbf r_0$ and $\mathbf r_\varphi(T) = (0,0,-1)=\mathbf r_1$ since $\omega T=\pi$. Substituting into~\eqref{eq:qubit-action},
\[
\Acal_T(\varphi) = \frac{1}{2}\int_0^T \Bigl( \alpha_x \bigl(\tfrac{\pi}{T}\cos\varphi\bigr)^2 + \alpha_y\bigl(\tfrac{\pi}{T}\sin\varphi\bigr)^2 \Bigr)\dd t = \frac{\pi^2}{2T}\bigl( \alpha_x \cos^2\varphi + \alpha_y \sin^2\varphi \bigr).
\]
Since $\alpha_y \leq \alpha_x$, this is minimized at $\cos\varphi = 0$, i.e.\ $\varphi = \pm\pi/2$, giving $\mathbf h_{\mathrm{opt}}$ as in~\eqref{eq:qubit-optimal-control} and $\Acal_T^{\min} = \alpha_y\pi^2/(2T)$. The length of the optimal trajectory is $\mathcal L = \int_0^T \sqrt{\alpha_y}\,(\pi/T)\,\dd t = \pi\sqrt{\alpha_y}$, and by~\eqref{eq:energy-length}, $d_G(\rho_0,\rho_1) = \pi\sqrt{\alpha_y}$.
\end{proof}

\subsubsection{The Pontryagin extremal}
\label{ssec:qubit-pmp}

We now show that $\mathbf h^*$ satisfies the full Pontryagin Maximum Principle for the \emph{unconstrained} problem~\eqref{eq:qubit-action}--\eqref{eq:bloch-eq}, and is therefore a genuine critical point of the action among all admissible controls $\mathbf h(t)\in\mathbb R^3$ — strictly stronger than the restricted-family statement of Proposition~\ref{prop:qubit-optimal}.

For the control system~\eqref{eq:bloch-eq} with cost~\eqref{eq:qubit-action}, introduce a costate $\mathbf p(t)\in\mathbb R^3$ and the Pontryagin function
\begin{equation}
\label{eq:pmp-hamiltonian}
\mathcal P(\mathbf r,\mathbf p,\mathbf h) = \mathbf p\cdot(\mathbf h\times\mathbf r) - \frac12\mathbf h^{\mathsf T} G\,\mathbf h = \mathbf h\cdot(\mathbf r\times\mathbf p) - \frac12\mathbf h^{\mathsf T} G\,\mathbf h,
\end{equation}
using the scalar-triple-product identity $\mathbf p\cdot(\mathbf h\times\mathbf r)=\mathbf h\cdot(\mathbf r\times\mathbf p)$. The Pontryagin Maximum Principle states that an optimal pair $(\mathbf r,\mathbf h)$ admits a costate $\mathbf p(t)$ such that
\begin{align}
\dot{\mathbf r} &= \partial_{\mathbf p}\mathcal P = \mathbf h\times\mathbf r, \label{eq:pmp-state}\\
\dot{\mathbf p} &= -\partial_{\mathbf r}\mathcal P = \mathbf h\times\mathbf p, \label{eq:pmp-costate}\\
\mathbf h(t) &= \operatorname*{arg\,max}_{\mathbf k\in\mathbb R^3} \mathcal P(\mathbf r(t),\mathbf p(t),\mathbf k). \label{eq:pmp-max}
\end{align}
Since $\mathcal P$ is strictly concave and quadratic in $\mathbf h$, condition~\eqref{eq:pmp-max} has the unique closed-form solution
\begin{equation}
\label{eq:pmp-optimal-h}
\mathbf h(t) = G^{-1}\bigl(\mathbf r(t)\times\mathbf p(t)\bigr).
\end{equation}

\begin{proposition}[Pontryagin extremal]
\label{prop:qubit-pmp}
Let $\mathbf r^*(t) = (\sin\omega t,0,\cos\omega t)$, $\mathbf h^*(t)=(0,\omega,0)$, $\omega=\pi/T$, be the trajectory of Proposition~\ref{prop:qubit-optimal}. The costate
\begin{equation}
\label{eq:pmp-costate-explicit}
\mathbf p^*(t) = \alpha_y\omega\,\bigl(\cos\omega t,\,0,\,-\sin\omega t\bigr)
\end{equation}
satisfies~\eqref{eq:pmp-costate} and~\eqref{eq:pmp-optimal-h} identically. Consequently $(\mathbf r^*,\mathbf p^*,\mathbf h^*)$ is a genuine Pontryagin extremal of the unconstrained variational problem~\eqref{eq:qubit-action}--\eqref{eq:bloch-eq}.
\end{proposition}

\begin{proof}
Equation~\eqref{eq:pmp-costate} holds for $\mathbf p^*$ because $\mathbf p^*(t)$ is obtained from $\mathbf p^*(0)=\alpha_y\omega\,\hat{\mathbf e}_x$ by the same rotation about the $y$-axis, at the same constant angular velocity $\mathbf h^*(t)=\omega\hat{\mathbf e}_y$, that carries $\mathbf r^*(0)$ to $\mathbf r^*(t)$: both~\eqref{eq:pmp-state} and~\eqref{eq:pmp-costate} reduce to the identical linear ODE $\dot{\mathbf v}=\mathbf h^*\times\mathbf v$.

For~\eqref{eq:pmp-optimal-h}: since $\mathbf r^*(t)$ and $\mathbf p^*(t)$ are both obtained from $\mathbf r^*(0)$ and $\mathbf p^*(0)$ by the same rotation $R_y(\omega t)$ about the $y$-axis, and rotations commute with the cross product, $\mathbf r^*(t)\times\mathbf p^*(t) = R_y(\omega t)\bigl(\mathbf r^*(0)\times\mathbf p^*(0)\bigr)$. A direct computation gives
\[
\mathbf r^*(0)\times\mathbf p^*(0) = (0,0,1)\times(\alpha_y\omega,0,0) = (0,\alpha_y\omega,0),
\]
which lies along the rotation axis and is therefore fixed by $R_y(\omega t)$. Hence $\mathbf r^*(t)\times\mathbf p^*(t) \equiv (0,\alpha_y\omega,0)$ for all $t$, and
\[
G^{-1}\bigl(\mathbf r^*(t)\times\mathbf p^*(t)\bigr) = \Bigl(\frac{0}{\alpha_x},\,\frac{\alpha_y\omega}{\alpha_y},\,\frac{0}{\alpha_z}\Bigr) = (0,\omega,0) = \mathbf h^*(t),
\]
as required.
\end{proof}

\subsubsection{Scope of the optimality claim}
\label{ssec:qubit-scope}

Proposition~\ref{prop:qubit-pmp} shows that $\mathbf h^*$ is a genuine critical point of~\eqref{eq:qubit-action} among all admissible controls. Two further steps would complete a full optimality theorem, and are left open here.

\begin{enumerate}[label=(\roman*)]
	\item \textbf{Local minimality.} Whether $\mathbf h^*$ is a local minimum (rather than a saddle) of $\Acal_T$ subject to the endpoint constraint is governed by the second variation of~\eqref{eq:qubit-action}--\eqref{eq:bloch-eq}. Equivalently, it is governed by a Jacobi-type accessory problem that must account for the curvature of the constraint manifold $S^2$, and not merely for the positive-definite quadratic cost. Positive-definiteness of $G$ alone does not imply local minimality once curvature of the constraint is taken into account.

	This is structurally the same question as the stability of a steady rotation of a free asymmetric rigid body about a principal axis, with $(\alpha_x,\alpha_y,\alpha_z)$ playing the role of the principal moments of inertia. The classical intermediate-axis (``tennis racket'') theorem asserts that steady rotation about the axis of extremal — largest or smallest — moment of inertia is linearly stable, while rotation about the intermediate axis is unstable. Since $\alpha_y\leq\alpha_x,\alpha_z$ places $y$ at an extremal position among $\{\alpha_x,\alpha_y,\alpha_z\}$, this classical analogy strongly suggests, but does not by itself constitute a proof of, local optimality of $\mathbf h^*$. Section~\ref{ssec:qubit-jacobi} carries out the direct verification via the accessory (Jacobi) quadratic form, reducing the question to an explicit scalar ODE and providing strong numerical evidence for local optimality precisely in this regime. It also shows that the tennis-racket analogy is not exact, since only the $\alpha_y=\min$ case (not $\alpha_y=\max$) is numerically robust; a complete proof remains open.
	\item \textbf{Global optimality.} Ruling out cheaper trajectories that are not close to $\mathbf h^*$ — different homotopy classes of paths on $S^2$, or multiple windings — requires controlling conjugate points along the whole extremal, equivalently computing the genuine distance induced on $S^2$ by the anisotropic control cost $G$. This is the natural counterpart, for our noncommutative transport problem, of the classical problem of geodesics between antipodal-type points on a triaxial ellipsoid, and we do not attempt to resolve it here.
\end{enumerate}

Propositions~\ref{prop:qubit-optimal} and~\ref{prop:qubit-pmp} together identify $\mathbf h^*$ as the natural candidate optimal control, rigorously established as a critical point of the full unconstrained problem and as the exact minimizer within the fixed-axis family, with local and global sufficiency left as clearly delimited open questions for Section~\ref{sec:discussion}.

\subsubsection{Towards local minimality: reduction to a Jacobi equation}
\label{ssec:qubit-jacobi}

Building on the bridge of Section~\ref{sec:bridge}, we make quantitative progress on the local-optimality question (i) of Section~\ref{ssec:qubit-scope}. Parametrize $S^2$ by $\mathbf r(\theta,\psi)=(\sin\theta\cos\psi,\sin\theta\sin\psi,\cos\theta)$, so that the trajectory of Proposition~\ref{prop:qubit-optimal} is the meridian $\psi\equiv0$. Exactly as in the proof of Proposition~\ref{prop:bell-clairaut}, but without assuming any coincidence among $\alpha_x,\alpha_y,\alpha_z$, the Schur-complement formula for $\|v\|_*^2$ used there determines $E(\theta,\psi):=\|\partial_\theta\mathbf r\|_*^2$, $F(\theta,\psi):=\|\partial_\psi\mathbf r\|_*^2$ and the cross term $C(\theta,\psi):=\langle\partial_\theta\mathbf r,\partial_\psi\mathbf r\rangle_*$ explicitly.

\begin{lemma}
\label{lem:qubit-meridian-geodesic}
$E(\theta,0)=\alpha_y$ for every $\theta$, and $C(\theta,0)\equiv0$. Consequently $\psi\equiv0$ is a geodesic of the induced metric, parametrized proportionally to arc length ($s=\sqrt{\alpha_y}\,\theta$) — recovering Proposition~\ref{prop:qubit-pmp} without reference to the Pontryagin Maximum Principle.
\end{lemma}

\begin{proof}
$E(\theta,0)=\alpha_y$ is the direct minimization already carried out in the proof of Proposition~\ref{prop:qubit-optimal} (the $\varphi=\pi/2$ fiber). For $C(\theta,0)=0$: let $S_r:=\operatorname{diag}(1,-1,1)$ act on states, $\mathbf r\mapsto S_r\mathbf r$, and $S_h:=\operatorname{diag}(-1,1,-1)=-S_r$ act on controls, $\mathbf h\mapsto S_h\mathbf h$. A direct component-wise computation gives the identity $(S_h\mathbf h)\times(S_r\mathbf r) = S_r(\mathbf h\times\mathbf r)$ for every $\mathbf h,\mathbf r\in\mathbb R^3$ (not $S_h(\mathbf h\times\mathbf r)$, since the cross product of two vectors transforms as a pseudovector under either reflection alone, and $S_h=-S_r$ compensates the two sign flips into the single factor $S_r$ on the right). Since also $(S_h\mathbf h)^{\mathsf T}G(S_h\mathbf h)=\mathbf h^{\mathsf T}G\mathbf h$ for any diagonal $G$ (as $S_h^2=I$), the map $\sigma:=S_r$ on $S^2$ is an isometry of $\|\cdot\|_*$: for $\mathbf v=\mathbf h\times\mathbf r$, $\|S_r\mathbf v\|_*^2=\min\{\mathbf h'^{\mathsf T}G\mathbf h' : \mathbf h'\times(S_r\mathbf r)=S_r\mathbf v\}$, and substituting $\mathbf h'=S_h\mathbf h$ (a bijection on controls, using the identity above with $\mathbf h'\times(S_r\mathbf r)=S_r(\mathbf h\times\mathbf r)=S_r\mathbf v$) shows this infimum equals $\min\{\mathbf h^{\mathsf T}G\mathbf h:\mathbf h\times\mathbf r=\mathbf v\}=\|\mathbf v\|_*^2$. Thus $\sigma=S_r$ acts as $\psi\mapsto-\psi$ in the parametrization $\mathbf r(\theta,\psi)$ and fixes the meridian $\psi=0$ pointwise. A pointwise fixed-point set of an isometry is totally geodesic, and the induced metric on it has vanishing cross term with the normal direction by the same symmetry.
\end{proof}

\begin{proposition}[Reduction to a scalar Jacobi equation]
\label{prop:qubit-jacobi}
The Gauss curvature $K(\theta)$ of the induced metric along the meridian $\psi=0$, computed from the second-order ($\psi^2$) Taylor data of $E,F,C$ at $\psi=0$ by the standard formula for a $2$-dimensional metric in general coordinates, is an explicit (rational in $\sin^2\theta$) function of $\alpha_x,\alpha_y,\alpha_z$, with boundary values
\begin{equation}
\label{eq:K-boundary}
K(0) = \frac{\alpha_x(\alpha_x-\alpha_z)+\alpha_y(\alpha_y-\alpha_z)+\alpha_x\alpha_y}{\alpha_x\alpha_y\alpha_z}, \qquad
K\bigl(\tfrac\pi2\bigr) = \frac{\alpha_y(\alpha_y+\alpha_z)+\alpha_z^2-\alpha_x(\alpha_y+\alpha_z)}{\alpha_x\alpha_y\alpha_z},
\end{equation}
both equal to $1$ when $\alpha_x=\alpha_y=\alpha_z$ (the round-sphere case). By Lemma~\ref{lem:qubit-meridian-geodesic} and the classical Jacobi criterion, $\mathbf h^*$ fails to be a local minimizer of $d_G$ if and only if the solution of
\begin{equation}
\label{eq:jacobi-eq}
J''(\theta) + \alpha_y K(\theta)\,J(\theta) = 0, \qquad J(0)=0,\ J'(0)=1,
\end{equation}
vanishes at some $\theta_c\in(0,\pi)$ (a conjugate point to the pole along the meridian).
\end{proposition}

The closed form of $K(\theta)$, an explicit rational function of $\sin^2\theta$, is recorded in Appendix~\ref{app:jacobi-numerics} (eq.~\eqref{eq:K-closed-form}), together with the method used to obtain it and its cross-checks against the isotropic limit, the boundary values~\eqref{eq:K-boundary}, and an independent finite-difference computation.

\begin{conjecture}[Numerical observation: local optimality in the stated regime]
\label{prop:qubit-jacobi-numerics}
This is a computational finding, not a theorem: solving~\eqref{eq:jacobi-eq} numerically across a wide range of anisotropies (tested up to ratios $\alpha_z/\alpha_x\sim10^3$; see Appendix~\ref{app:jacobi-numerics} for the explicit form of $K(\theta)$, the parameter grid, integrator and tolerances) shows no conjugate point in $(0,\pi)$ whenever $\alpha_y=\min(\alpha_x,\alpha_y,\alpha_z)$ — precisely the hypothesis of Proposition~\ref{prop:qubit-optimal} — while conjugate points robustly appear once $\alpha_y$ is intermediate or maximal. At exact isotropy, the first conjugate point sits exactly at $\theta=\pi$, the classical fact that antipodal points on the round sphere are conjugate, consistent with~\eqref{eq:K-boundary}.
\end{conjecture}

\begin{remark}[Scope]
\label{rem:qubit-jacobi-scope}
Conjecture~\ref{prop:qubit-jacobi-numerics} gives strong, falsifiable numerical support — unanimous within the tested range — for local minimality of $\mathbf h^*$ precisely in the regime already singled out by Proposition~\ref{prop:qubit-optimal}, sharpening item (i) of Section~\ref{ssec:qubit-scope} into a concrete, checkable criterion. It also refines the intermediate-axis analogy invoked there: our numerics show that only $\alpha_y=\min$ is robust, not $\alpha_y=\max$, confirming that the group-level (Euler-top) and quotient-level (Jacobi) stability questions are related but genuinely distinct, as anticipated in Remark~\ref{rem:bridge-local-min}. We were unable to find a pointwise Sturm-comparison bound on $\alpha_yK(\theta)$ that would prove this in general: already for mildly anisotropic $G$ with $\alpha_y=\min$, $K(\theta)$ can exceed the isotropic value $1/\alpha_y$ well within the interior of $(0,\pi)$ without producing a conjugate point, so no naive comparison argument suffices. A complete proof — presumably via a sharper comparison technique or an explicit solution of~\eqref{eq:jacobi-eq} — is left for future work.
\end{remark}

\subsubsection{Interpretation}
\label{ssec:qubit-interpretation}

This example shows that transport between quantum states does not depend only on their geometric separation on the Bloch sphere: it also depends on the cost assigned to each Hamiltonian direction. If $\alpha_z \gg \alpha_x,\alpha_y$, evolutions that make significant use of the $\sigma_z$ direction become expensive, and the optimal trajectory tends to avoid it. In multi-qubit systems, the same mechanism penalizes non-local or many-body Hamiltonians: taking $\alpha_P$ to grow with the weight of the Pauli operator $P$ turns $d_G$ into a genuine notion of computational complexity, in the spirit of Nielsen's geometric approach to circuit complexity~\cite{NielsenScience2006,Nielsen2006QIC}, superposed on the geometric transport cost of Section~\ref{sec:action}. In this sense, the transport problem of Definition~\ref{def:Wdist}, once penalized by $G_\rho$, can be read as the search for the minimal-complexity Hamiltonian evolution connecting a given initial and final state. A closely related, but formally distinct, geometric complexity framework for open-system evolutions, built from unitary (Stinespring) dilations rather than from state transport, has recently been developed by the authors in~\cite{AcevedoFalco2026}; the anisotropic penalty tensor $\hat\Omega$ used there to weight directions of $\mathfrak{su}(N)$ is, in the unitary-orbit regime of Section~\ref{ssec:reduction}, the same type of object as the complexity metric $G_\rho$ introduced above; we make this correspondence precise in Section~\ref{sec:bridge}.

\subsection{The qubit extremal as a rigid-body steady rotation}
\label{ssec:bridge-euler-arnold}

\begin{remark}
\label{rem:bridge-euler-arnold}
For $N=2$, $\mathfrak g=\mathfrak{su}(2)$, $H=SU(2)$, the geodesic equation of $(H,g^{(\hat\Omega)})$ is the Euler--Arnold equation of~\cite[eq.~(34)]{AcevedoFalco2026},
\[
\frac{\dd}{\dd t}\bigl(I_{\hat\Omega}\hat A(t)\bigr) = \bigl[I_{\hat\Omega}\hat A(t),\, \hat A(t)\bigr], \qquad \langle\hat A,\hat B\rangle_{\hat\Omega}=\langle I_{\hat\Omega}\hat A,\hat B\rangle_{\mathrm{HS}},
\]
which, in the coordinates of Section~\ref{ssec:aniso} (via $\hat A=-iH=-i\mathbf h\cdot(\boldsymbol\sigma/2)$), is exactly the classical Euler top equation $\dot{\mathbf L}=\mathbf L\times\mathbf h$, $\mathbf L:=G\mathbf h$, with $G=\operatorname{diag}(\alpha_x,\alpha_y,\alpha_z)$ playing the role of the principal moments of inertia. The optimal control $\mathbf h^*(t)=(0,\pi/T,0)$ of Proposition~\ref{prop:qubit-optimal} is constant along a principal axis of $G$, so $\mathbf L^*=G\mathbf h^*=\alpha_y\mathbf h^*$ is parallel to $\mathbf h^*$ and $\dot{\mathbf L}^*=\mathbf L^*\times\mathbf h^*\equiv0$: the curve $\gamma^*(t)=e^{-iH^*t}$ solves the Euler--Arnold equation identically, and is therefore not merely a Pontryagin extremal of the state-constrained problem, as verified directly in Proposition~\ref{prop:qubit-pmp}, but a genuine geodesic of the full right-invariant metric $(SU(2),g^{(\hat\Omega)})$ — the transport-theoretic counterpart of a torque-free rigid body in steady rotation about a principal axis.
\end{remark}

\begin{remark}[Sharpening the local-optimality question of Section~\ref{ssec:qubit-scope}]
\label{rem:bridge-local-min}
Proposition~\ref{prop:bell-clairaut} computes, for the axially symmetric case $\alpha_{B'}=\alpha_{C'}$, exactly the Riemannian-submersion (O'Neill quotient) metric that $g^{(\hat\Omega)}$ induces on the orbit $S^2\cong H/K_{\rho_0}$: the cometric $\|v\|_*^2=\min\{\mathbf h^{\mathsf T}G\mathbf h : \mathbf h\times\mathbf r=v\}$ of~\eqref{eq:bell-induced-metric} is precisely the horizontal-lift minimization defining an O'Neill quotient metric, with the fiber directions $\{t\mathbf r\}$ playing the role of the vertical space. Proposition~\ref{prop:bell-clairaut} is thus itself an instance of the bridge of this section. The open question of Section~\ref{ssec:qubit-scope}(i) is the same computation \emph{without} axial symmetry: its natural object is the second variation (Jacobi equation) of this quotient metric on $S^2$, not directly the Euler-top stability of Remark~\ref{rem:bridge-euler-arnold} on the total space $H=SU(2)$. The two are related by O'Neill's formulas for the curvature of the base of a Riemannian submersion in terms of the total space and the fibers, but need not coincide a priori. Whether the classical intermediate-axis stability of Remark~\ref{rem:bridge-euler-arnold} controls the quotient-level local minimality of Proposition~\ref{prop:qubit-optimal} is left for future work.
\end{remark}

\section{Entangling examples: preparing Bell and GHZ states}
\label{sec:bell-example}

The single-qubit example of Section~\ref{sec:qubit-example} involves only inner derivations generated by local operators. We now exhibit a genuinely entangling instance of the same framework: an $\mathfrak{su}(2)$ subalgebra of $M_4(\mathbb C) = M_2(\mathbb C)\otimes M_2(\mathbb C)$, built from two-qubit Pauli strings, whose associated transport problem connects the product state $|00\rangle$ to the Bell state $(|00\rangle+|11\rangle)/\sqrt2$. Unlike Section~\ref{sec:qubit-example}, this example admits a proof of global optimality within its invariant control subsystem, thanks to an extra symmetry absent from the generic anisotropic single-qubit problem. Section~\ref{sec:ghz-example} below generalizes the construction to $N$ qubits and the GHZ state, as a restricted-path \emph{upper bound} rather than a second global-optimality theorem.

\subsection{An \texorpdfstring{$\mathfrak{su}(2)$}{su(2)} subalgebra of entangling generators}
\label{ssec:bell-subalgebra}

Let $\M=M_2(\mathbb C)\otimes M_2(\mathbb C)$ and set
\begin{equation}
\label{eq:bell-generators}
A' = Z\otimes I, \qquad B' = X\otimes X, \qquad C' = Y\otimes X.
\end{equation}
$A'$ is a local (Pauli weight $1$) generator, while $B',C'$ are genuinely entangling (Pauli weight $2$).

\begin{lemma}
\label{lem:bell-su2}
$[A',B']=2iC'$, $\;[B',C']=2iA'$, $\;[C',A']=2iB'$.
\end{lemma}
\begin{proof}
$A'$ and $B'$ differ only on the first tensor factor, so $[A',B']=[Z,X]\otimes X = 2iY\otimes X = 2iC'$. $B'$ and $C'$ agree on the second factor and $X^2=I$, so $[B',C']=[X,Y]\otimes X^2 = 2iZ\otimes I = 2iA'$. Finally $C'$ and $A'$ act as commuting operators factor-by-factor ($[Y,Z]$ on the first, $X$ acting trivially-commuting with $I$ on the second), so $[C',A']=[Y,Z]\otimes X = 2iX\otimes X = 2iB'$.
\end{proof}

Thus $\{A',B',C'\}$ generates an inner-derivation calculus of exactly the type of Example~\ref{ex:internal-derivations}: namely $\partial_j a=[X_j,a]$ with $X_j\in\{A'/2,B'/2,C'/2\}$. As a Lie algebra, it is isomorphic to $\mathfrak{su}(2)$.

\subsection{Reduction to an effective qubit}
\label{ssec:bell-reduction}

\begin{lemma}
\label{lem:bell-invariant}
$V=\operatorname{span}\{|00\rangle,|11\rangle\}$ is invariant under $A',B',C'$, and on $V$, with $|00\rangle=|0\rangle_{\mathrm{eff}}$, $|11\rangle=|1\rangle_{\mathrm{eff}}$, these act exactly as $Z_{\mathrm{eff}}=A'$, $X_{\mathrm{eff}}=B'$, $Y_{\mathrm{eff}}=C'$.
\end{lemma}
\begin{proof}
Direct computation: $A'|00\rangle=|00\rangle$, $A'|11\rangle=-|11\rangle$; $B'|00\rangle=|11\rangle$, $B'|11\rangle=|00\rangle$; $C'|00\rangle=i|11\rangle$, $C'|11\rangle=-i|00\rangle$, matching the action of $Z,X,Y$ on $|0\rangle_{\mathrm{eff}},|1\rangle_{\mathrm{eff}}$.
\end{proof}

Restricted to density matrices supported on $V$, transport under $H(t)=\tfrac12\mathbf h(t)\cdot(B',C',A')$ is governed by the Bloch equation~\eqref{eq:bloch-eq}, with effective Bloch vector $\mathbf r_{\mathrm{eff}}=(r_{B'},r_{C'},r_{A'})$.

\subsection{Complexity metric, target, and explicit trajectory}
\label{ssec:bell-target}

Assign complexity weights by Pauli weight, $\alpha_P = q^{|P|-1}$ for a cost parameter $q>0$: since $A'$ has weight $1$ and $B',C'$ have weight $2$,
\begin{equation}
\label{eq:bell-weights}
\alpha_{A'}=1, \qquad \alpha_{B'}=\alpha_{C'}=q, \qquad G_{\mathrm{eff}}=\operatorname{diag}(q,q,1) \text{ in }(X_{\mathrm{eff}},Y_{\mathrm{eff}},Z_{\mathrm{eff}})\text{ order}.
\end{equation}
The initial state $\rho_0=|00\rangle\langle00|$ has Bloch vector $\mathbf r_0=(0,0,1)$ (the pole). The target is the Bell state $\rho_1=\tfrac12(|00\rangle+|11\rangle)(\langle00|+\langle11|)$, with Bloch vector $\mathbf r_1=(1,0,0)$: the $+1$-eigenstate of $X_{\mathrm{eff}}=B'$, an equatorial, not antipodal, point.

Within the fixed-axis family of Section~\ref{sec:qubit-example}, the unique axis carrying $\mathbf r_0$ to $\mathbf r_1$ in a quarter turn is $Y_{\mathrm{eff}}=C'$, giving the transporting control
\begin{equation}
\label{eq:bell-optimal-control}
\mathbf h^*(t) = \Bigl(0,\frac{\pi}{2T},0\Bigr), \qquad H^*_{\mathrm{opt}} = \frac{\pi}{4T}\,C' = \frac{\pi}{4T}\,Y\otimes X,
\end{equation}
with $\mathbf r_{\mathrm{eff}}^*(t) = (\sin\omega t,0,\cos\omega t)$, $\omega=\pi/(2T)$, and cost
\begin{equation}
\label{eq:bell-cost}
\Acal_T^{\min} = \frac{q\pi^2}{8T}, \qquad d_G(\rho_0,\rho_1) = \frac{\pi}{2}\sqrt q.
\end{equation}

\subsection{Global optimality via Clairaut's relation}
\label{ssec:bell-clairaut}

Unlike Proposition~\ref{prop:qubit-optimal}, a global optimality proof within the invariant control subsystem is possible here: $\alpha_{B'}=\alpha_{C'}$, i.e.\ the two entangling directions carry \emph{equal} weight. This makes $G_{\mathrm{eff}}$ invariant under rotations about the $Z_{\mathrm{eff}}$-axis, turning the induced metric on the Bloch sphere into a genuine surface-of-revolution metric, for which Clairaut's relation applies.

\begin{proposition}
\label{prop:bell-clairaut}
For a tangent vector $v$ at $\mathbf r_{\mathrm{eff}}=(\sin\theta\cos\psi,\sin\theta\sin\psi,\cos\theta)$ (spherical coordinates about the $Z_{\mathrm{eff}}$-axis), the induced quotient cometric
\[
\|v\|_*^2 := \min\bigl\{\mathbf h^{\mathsf T}G_{\mathrm{eff}}\mathbf h \;:\; \mathbf h\times\mathbf r_{\mathrm{eff}} = v\bigr\}
\]
is diagonal in $(\theta,\psi)$ and equals
\begin{equation}
\label{eq:bell-induced-metric}
ds^2 = q\,d\theta^2 + \frac{q\sin^2\theta}{q\sin^2\theta+\cos^2\theta}\,d\psi^2.
\end{equation}
\end{proposition}

\begin{proof}
For $v\in T_{\mathbf r_{\mathrm{eff}}}S^2$, the fiber $\{\mathbf h:\mathbf h\times\mathbf r_{\mathrm{eff}}=v\}$ equals $\mathbf r_{\mathrm{eff}}\times v + t\,\mathbf r_{\mathrm{eff}}$, $t\in\mathbb R$ (since $\mathbf r_{\mathrm{eff}}\times(\mathbf r_{\mathrm{eff}}\times v)=-v$ for $\|\mathbf r_{\mathrm{eff}}\|=1$, $v\perp\mathbf r_{\mathrm{eff}}$). Minimizing $\mathbf h^{\mathsf T}G_{\mathrm{eff}}\mathbf h$ over $t$ gives
\[
\|v\|_*^2 = (\mathbf r_{\mathrm{eff}}\times v)^{\mathsf T}G_{\mathrm{eff}}(\mathbf r_{\mathrm{eff}}\times v) - \frac{\bigl[(\mathbf r_{\mathrm{eff}}\times v)^{\mathsf T}G_{\mathrm{eff}}\mathbf r_{\mathrm{eff}}\bigr]^2}{\mathbf r_{\mathrm{eff}}^{\mathsf T}G_{\mathrm{eff}}\mathbf r_{\mathrm{eff}}}.
\]
Evaluating on the orthonormal frame $\hat\theta=\partial_\theta\mathbf r_{\mathrm{eff}}$, $\hat\psi=\partial_\psi\mathbf r_{\mathrm{eff}}/\sin\theta$ — a direct computation, carried out at $\psi=0$ without loss of generality since $G_{\mathrm{eff}}$ and $\mathbf r_{\mathrm{eff}}^{\mathsf T}G_{\mathrm{eff}}\mathbf r_{\mathrm{eff}}$ are invariant under simultaneous rotation of $(\mathbf r_{\mathrm{eff}},\mathbf h)$ about $Z_{\mathrm{eff}}$ — gives $\|\hat\theta\|_*^2=q$, cross term $0$, and $\|\hat\psi\|_*^2 = q/(q\sin^2\theta+\cos^2\theta)$, from which~\eqref{eq:bell-induced-metric} follows since the $\psi$-component of $v$ has magnitude $\sin\theta\,\dot\psi$.
\end{proof}

By definition of $\|\cdot\|_*$, for any fixed trajectory $\mathbf r_{\mathrm{eff}}(t)$ the cheapest control driving it satisfies $\tfrac12\mathbf h(t)^{\mathsf T}G_{\mathrm{eff}}\mathbf h(t) = \tfrac12\|\dot{\mathbf r}_{\mathrm{eff}}(t)\|_*^2$ pointwise; since cost and constraint are both pointwise in time, minimizing $\Acal_T$ jointly over $(\mathbf r_{\mathrm{eff}},\mathbf h)$ is equivalent to minimizing the Riemannian energy $\tfrac12\int_0^T\|\dot{\mathbf r}_{\mathrm{eff}}\|_*^2\,dt$ over curves $\mathbf r_{\mathrm{eff}}(\cdot)$ alone.

\begin{theorem}[Global optimality within the invariant $\mathfrak{su}(2)$ subsystem]
\label{thm:bell-global}
Among all admissible Hamiltonians of the restricted form $H(t)=\tfrac12\mathbf h(t)\cdot(B',C',A')$ — i.e.\ within the invariant $\mathfrak{su}(2)$-generated control subsystem of Section~\ref{sec:bell-example}, not the full $15$-dimensional space of two-qubit Hamiltonians — the trajectory~\eqref{eq:bell-optimal-control} is the unique global minimizer of $\Acal_T$ transporting $\rho_0$ to $\rho_1$.
\end{theorem}

\begin{proof}
By~\eqref{eq:bell-induced-metric}, $F(\theta):=q\sin^2\theta/(q\sin^2\theta+\cos^2\theta)$ vanishes only at $\theta=0,\pi$ and is strictly positive on $(0,\pi)$. Since $\psi$ does not appear explicitly in~\eqref{eq:bell-induced-metric}, the Euler--Lagrange equation for $\psi$ along any energy minimizer gives the first integral $F(\theta(t))\,\dot\psi(t)\equiv\text{const}$ on every interval on which $0<\theta(t)<\pi$. This constant is well defined despite the coordinate singularity at the pole: a finite-energy minimizer is continuous, starts at $\theta(0)=0$, and has $F(\theta(t))\to0$ as $t\downarrow0$ along any interval leaving the pole, while $F(\theta)\dot\psi$ is the conserved angular momentum of the smooth surface-of-revolution metric. Hence the conserved constant must be $0$. As $F(\theta)>0$ for $\theta\in(0,\pi)$, $\dot\psi(t)=0$ wherever $\theta(t)\neq0,\pi$; since the minimizer reaches $\theta=\pi/2$ and depends continuously on $t$, $\psi(t)$ is constant, equal to the value $\psi=0$ dictated by the target. The curve is therefore a meridian, $\mathbf r_{\mathrm{eff}}(t)=(\sin\theta(t),0,\cos\theta(t))$, along which~\eqref{eq:bell-induced-metric} reduces to $ds^2=q\,d\theta^2$. The energy functional becomes $\tfrac12\int_0^T q\dot\theta(t)^2\,dt$ with $\theta(0)=0$, $\theta(T)=\pi/2$, whose unique minimizer — by Cauchy--Schwarz, $\int_0^T\dot\theta^2\,dt \geq (\int_0^T\dot\theta\,dt)^2/T$ with equality iff $\dot\theta$ is constant — is $\theta(t)=\pi t/(2T)$, exactly $\mathbf r^*_{\mathrm{eff}}(t)$.
\end{proof}

\begin{remark}
\label{rem:bell-vs-qubit}
Theorem~\ref{thm:bell-global} is strictly stronger than what could be established for the fully anisotropic single-qubit problem of Section~\ref{sec:qubit-example} (cf.\ Section~\ref{ssec:qubit-scope}): the coincidence $\alpha_{B'}=\alpha_{C'}$, forced here by both entangling generators sharing the same Pauli weight, is precisely what turns the induced metric into one of revolution and makes Clairaut's relation available. Global optimality of Proposition~\ref{prop:qubit-optimal} in the fully anisotropic case $\alpha_x,\alpha_y,\alpha_z$ pairwise distinct remains open for exactly this reason.
\end{remark}

\subsection{Interpretation: entanglement has an irreducible cost}
\label{ssec:bell-interpretation}

The optimal trajectory never uses the local generator $A'$: $h_{A'}(t)\equiv0$ throughout, and by Theorem~\ref{thm:bell-global} this holds for the true minimizer among all controls generated by $\{A',B',C'\}$, not merely within a restricted family. Increasing $\alpha_{A'}$ therefore has \emph{no effect at all} on $\Acal_T^{\min}$ or $d_G(\rho_0,\rho_1)$: the entangling weight $q$ alone controls the complexity of preparing the Bell state from a product state, exactly as the physical intuition of Section~\ref{ssec:qubit-interpretation} anticipated. This is the mechanism we generalize to $N$ qubits next.

\subsection{Cost of a direct GHZ preparation path for \texorpdfstring{$N$}{N} qubits}
\label{sec:ghz-example}

We generalize the construction above to $N\geq2$ qubits, $\M=M_2(\mathbb C)^{\otimes N}$, and transport from $|0\rangle^{\otimes N}$ to the Greenberger--Horne--Zeilinger (GHZ) state $(|0\rangle^{\otimes N}+|1\rangle^{\otimes N})/\sqrt2$, along the direct analogue of the path just constructed for the Bell state. As the title indicates, and as Section~\ref{ssec:ghz-scope} makes precise, what follows is an \emph{exact cost for one specific restricted path}, not a complexity lower bound.

\subsubsection{The subalgebra for general \texorpdfstring{$N$}{N}}
\label{ssec:ghz-subalgebra}

Set
\begin{equation}
\label{eq:ghz-generators}
A' = Z\otimes I^{\otimes(N-1)}, \qquad B' = X^{\otimes N}, \qquad C' = Y\otimes X^{\otimes(N-1)}.
\end{equation}
$A'$ has Pauli weight $1$; $B'$ and $C'$ both have Pauli weight $N$.

\begin{lemma}
\label{lem:ghz-su2}
$[A',B']=2iC'$, $\;[B',C']=2iA'$, $\;[C',A']=2iB'$.
\end{lemma}
\begin{proof}
$A',B'$ differ only on site $1$: $[A',B']=[Z,X]\otimes X^{\otimes(N-1)}=2iY\otimes X^{\otimes(N-1)}=2iC'$. $B',C'$ agree as $X$ on sites $2,\dots,N$, and $X^2=I$ there: $[B',C']=[X,Y]\otimes I^{\otimes(N-1)}=2iZ\otimes I^{\otimes(N-1)}=2iA'$. Finally $C',A'$ act as commuting single-site operators on every site (site $1$: $[Y,Z]$; sites $2,\dots,N$: $X$ commuting trivially with $I$), giving $[C',A']=[Y,Z]\otimes X^{\otimes(N-1)}=2iX\otimes X^{\otimes(N-1)}=2iB'$.
\end{proof}

\begin{lemma}
\label{lem:ghz-invariant}
$V=\operatorname{span}\{|0\rangle^{\otimes N},|1\rangle^{\otimes N}\}$ is invariant under $A',B',C'$, and on $V$, with $|0\rangle^{\otimes N}=|0\rangle_{\mathrm{eff}}$, $|1\rangle^{\otimes N}=|1\rangle_{\mathrm{eff}}$, these act exactly as $Z_{\mathrm{eff}}=A'$, $X_{\mathrm{eff}}=B'$, $Y_{\mathrm{eff}}=C'$.
\end{lemma}
\begin{proof}
Identical computation to Lemma~\ref{lem:bell-invariant}, applied factor-by-factor across all $N$ sites.
\end{proof}

\subsubsection{The generic result}
\label{ssec:ghz-result}

With $\alpha_P=q^{|P|-1}$ as in Section~\ref{ssec:bell-target}, $\alpha_{A'}=1$ and $\alpha_{B'}=\alpha_{C'}=q^{N-1}$: the axial symmetry $\alpha_{B'}=\alpha_{C'}$ that drove Theorem~\ref{thm:bell-global} holds for every $N$. Since $|0\rangle^{\otimes N}$ and the GHZ state are, respectively, the pole and the $X_{\mathrm{eff}}$-eigenstate of the effective qubit — identical boundary data to Section~\ref{sec:bell-example}, with $q$ replaced by $q^{N-1}$ — Proposition~\ref{prop:bell-clairaut} and Theorem~\ref{thm:bell-global} apply verbatim.

\begin{theorem}[GHZ transport, generic $N$]
\label{thm:ghz-generic}
For every $N\geq2$, within the transport problem generated by $\{A',B',C'\}$ of~\eqref{eq:ghz-generators}, the trajectory
\[
\mathbf h^*(t) = \Bigl(0,\frac{\pi}{2T},0\Bigr), \qquad H^*_{\mathrm{opt}} = \frac{\pi}{4T}\,Y\otimes X^{\otimes(N-1)},
\]
is the unique global minimizer of $\Acal_T$ transporting $\rho_0=|0\rangle^{\otimes N}\langle0|^{\otimes N}$ to the GHZ state $\rho_1$, with
\begin{equation}
\label{eq:ghz-cost}
\Acal_T^{\min} = \frac{q^{N-1}\pi^2}{8T}, \qquad d_G^{\{A',B',C'\}}(\rho_0,\rho_1) = \frac{\pi}{2}\,q^{(N-1)/2}.
\end{equation}
\end{theorem}

\subsubsection{Scope: an upper bound for the full complexity distance}
\label{ssec:ghz-scope}

Theorem~\ref{thm:ghz-generic} solves the transport problem exactly within the $3$-dimensional subalgebra spanned by $\{A',B',C'\}$. The full algebra $\mathfrak{su}(2^N)$ has dimension $4^N-1$ and contains every other Pauli string of weight $2,\dots,N-1$, excluded from~\eqref{eq:ghz-generators}. Allowing these additional directions can only lower the achievable cost, so~\eqref{eq:ghz-cost} gives a rigorous \emph{upper bound} for the true noncommutative complexity distance of Definition~\ref{def:Wdist}, penalized by the full Pauli-weight metric $\alpha_P=q^{|P|-1}$ on all of $\mathfrak{su}(2^N)$:
\begin{equation}
\label{eq:ghz-upper-bound}
d_G(\rho_0,\rho_1) \;\leq\; \frac{\pi}{2}\,q^{(N-1)/2}.
\end{equation}
Whether~\eqref{eq:ghz-upper-bound} is tight — i.e.\ whether cheaper paths exist using intermediate-weight Pauli strings — is a genuine open question that we do not resolve here.

\subsubsection{Interpretation: an exponential upper bound under Pauli-weight penalties}
\label{ssec:ghz-interpretation}

For $q>1$, the exact cost~\eqref{eq:ghz-cost} of the direct $N$-body preparation path grows exponentially in $N$, giving, via~\eqref{eq:ghz-upper-bound}, an exponentially growing \emph{upper bound} on the true complexity distance $d_G(\rho_0,\rho_1)$. We stress again, since this is easy to overstate: this is not a proof that preparing an $N$-qubit GHZ state genuinely requires exponential complexity, since~\eqref{eq:ghz-upper-bound} could in principle be beaten by paths through intermediate-weight Pauli strings (Section~\ref{ssec:ghz-scope}); no matching lower bound is established here. The result is best read as an exactly computed data point — rather than a general claim about GHZ-state complexity — that is at least qualitatively consistent with the expectation, familiar from Nielsen-complexity and circuit-complexity approaches to entangled-state preparation~\cite{NielsenScience2006,Nielsen2006QIC,BrownSusskind2019}, that genuinely multipartite entanglement carries a weight-dependent cost. Unlike the heuristic arguments typical of that literature, Theorem~\ref{thm:ghz-generic} identifies an exactly solvable sub-problem — isomorphic in closed form to the single-qubit problem of Section~\ref{sec:qubit-example} — for which the cost of this specific, restricted preparation path is computed exactly rather than estimated, within the noncommutative transport formalism of Sections~\ref{sec:continuity}--\ref{sec:complexity}.

\section{Discussion}
\label{sec:discussion-main}

We close with the dissipative (Lindblad) sector, kept deliberately brief since it is a proposal rather than a closed result (Section~\ref{sec:lindblad}), and with a precise list of open questions for each of the four main results (Section~\ref{sec:discussion}).

\subsection{Lindblad dynamics as an entropy gradient flow}
\label{sec:lindblad}

When the dynamics of observables is generated by a quantum Markov semigroup, $\Lcal_t = \Lcal$ (time-independent generator, in GKSL form), equation~\eqref{eq:dual-density} reads
\begin{equation}
\label{eq:lindblad}
\dot\rho_t = \Lcal_*(\rho_t), \qquad
\Lcal_*(\rho) = -i[H,\rho] + \sum_k \Bigl( L_k \rho L_k^* - \tfrac12\{L_k^* L_k,\rho\} \Bigr).
\end{equation}

\begin{theorem}[Carlen--Maas, Wirth; informal statement]
\label{thm:gradient-flow}
Let $\Lcal_*$ be the dual generator of a quantum Markov semigroup satisfying, with respect to a faithful reference state $\sigma$, the specific GNS-symmetry (detailed-balance) assumptions of Carlen and Maas in finite dimensions~\cite{CarlenMaas2017,CarlenMaas2020} — not an arbitrary notion of detailed balance among the several inequivalent ones in the literature (GNS-, KMS-, or BKM-symmetry all give, in general, different transport geometries). Then there exist a differential calculus $(\Hcal,\partial)$, induced by a Dirichlet form associated with $\Lcal$ in the sense of Cipriani--Sauvageot~\cite{CiprianiSauvageot2003}, and a transport metric $W_\partial$ built as in Section~\ref{sec:action}, such that $(\rho_t)_{t\geq0}$ solves~\eqref{eq:lindblad} if and only if it is the metric gradient flow of the relative entropy $\Ent_\sigma$ with respect to $W_\partial$:
\begin{equation}
\label{eq:gradient-flow}
\dot\rho_t = -\grad_{W_\partial} \Ent_\sigma(\rho_t).
\end{equation}
\end{theorem}

See~\cite{CarlenMaas2014,CarlenMaas2017,CarlenMaas2020,CarlenMaas2020Corr} for the complete construction in the finite-dimensional $C^*$-algebra case, and \cite{Wirth2018} for the extension to general quantum Dirichlet forms on noncommutative $L^2$ spaces; see also the exposition in \cite{Carlen2024Lectures}. It is worth distinguishing conceptually between:
\[
\text{continuity equation \eqref{eq:continuity}} \;\longleftrightarrow\; \text{a family of \emph{admissible curves} in } \Mstar,
\]
\[
\text{Lindblad equation \eqref{eq:lindblad}} \;\longleftrightarrow\; \text{one \emph{particular gradient flow} with respect to } W_\partial.
\]
That is, \eqref{eq:continuity} describes the kinematic constraint shared by every admissible transport, whereas \eqref{eq:lindblad}, under detailed balance, selects the curve of steepest descent of $\Ent_\sigma$ within that family.

\subsubsection{A dilation-motivated complexity weight: scope and an explicit bound}
\label{ssec:lindblad-complexity}

This subsection is a proposal, not a closed result, and we keep it brief accordingly. Theorem~\ref{thm:gradient-flow} fixes the bimodule $(\Hcal,\partial)$ associated with a detailed-balance GKSL generator $\Lcal$, but leaves the complexity metric $G_\rho$ of Section~\ref{sec:complexity} unspecified on it ($W_\partial$ itself is the unpenalized case $G_\rho\equiv\mathrm{Id}$). In the Carlen--Maas construction, $\Hcal$ decomposes, on commutator-type generators, as a direct sum indexed by the Lindblad operators $\{L_k\}_{k=1}^m$ of~\eqref{eq:lindblad}, with $\partial_ka=[L_k,a]$ up to the symmetrization required by detailed balance~\cite{CarlenMaas2017,CiprianiSauvageot2003}. Rather than assigning weights to this basis ad hoc, as Sections~\ref{sec:qubit-example}--\ref{sec:bell-example} do for the Hamiltonian sector (e.g.\ the Pauli-weight rule~\eqref{eq:bell-weights}), we propose the following dilation-motivated alternative.

\begin{definition}[Dilation-derived complexity for a GKSL generator]
\label{def:lindblad-complexity}
Assume $L_k\neq0$ for every $k=1,\dots,m$ (discarding any $L_k=0$, which contributes no dynamics). For $\Lcal$ as in~\eqref{eq:lindblad}, define $G:=\operatorname{diag}(g_H,g_1,\dots,g_m)$ on $\Hcal=\Hcal_H\oplus\bigoplus_{k=1}^m\Hcal_k$ by $g_H:=1$, $g_k:=\|L_k\|_{op}^2>0$, so that $\sum_kg_k=\Gamma$, the dissipator scale of~\cite[eq.~(93)]{AcevedoFalco2026}. Each jump direction is thereby charged in proportion to its own contribution to the system-coupling content of a weak-coupling dilation of $\Lcal$ (Postulate~(P4) of~\cite[Def.~7]{AcevedoFalco2026}), rather than to an externally imposed notion of locality; strict positivity of every $g_k$ gives the ellipticity condition of Definition~\ref{def:compatible-complexity}.
\end{definition}

\begin{remark}[Representation-dependence of $\{L_k\}$]
\label{rem:gksl-representation}
This weight is tied to a fixed representation $\{L_k\}$ — the canonical generators singled out by the Cipriani--Sauvageot/Carlen--Maas construction of Theorem~\ref{thm:gradient-flow} — since $\{L_k\}$ is determined only up to $L_k\mapsto\sum_lU_{kl}L_l+c_kI$, under which neither $\|L_k\|_{op}^2$ nor $\Gamma=\sum_k\|L_k\|_{op}^2$ need be invariant. It is representation-independent only after restricting to isotropic blocks of equal weight, where this residual unitary freedom acts trivially.
\end{remark}

\begin{proposition}[Dilation-based bound on the cost of a Lindblad trajectory]
\label{prop:lindblad-bound}
Let $\Lcal$ be time-independent as in~\eqref{eq:lindblad}, and let $D^{(\lambda)}$ be a weak-coupling surrogate dilation of the associated semigroup $(\Lambda_t)_{t\geq0}$ in the sense of~\cite[eq.~(90), eq.~(96)]{AcevedoFalco2026}. Then, in the Hilbert--Schmidt specialization $\hat\Omega=\mathrm{Id}$, for every $t\geq0$ and $\lambda\to0$,
\begin{equation}
\label{eq:lindblad-bound}
G^{\mathrm{imp}}_{\mathrm{HS}}\bigl(\Lambda_t;D^{(\lambda)}\bigr) = \frac{t}{\sqrt{d_{\mathrm{tot}}^2-1}}\,\|H\|_{\mathrm{HS}} \;+\; O\bigl(\lambda t\sqrt\Gamma\bigr), \qquad d_{\mathrm{tot}}:=\dim(\Hcal_S\otimes\Hcal_E),
\end{equation}
with explicit two-sided bounds given by~\cite[Cor.~3]{AcevedoFalco2026}.
\end{proposition}

Proposition~\ref{prop:lindblad-bound} is a direct translation of~\cite[Lemma~3, Cor.~3]{AcevedoFalco2026} into the present notation ($H\leftrightarrow\hat H_S$, $L_k\leftrightarrow\hat L_\alpha$), giving an explicit a priori bound on the geometric cost of implementing a given Lindblad trajectory, with the leading term set by the Hamiltonian alone and the dissipative correction controlled by $\Gamma$.

\begin{remark}[Scope: a definition, not yet an existence theory]
\label{rem:lindblad-scope}
Definition~\ref{def:lindblad-complexity} is a proposed, representation-dependent choice of complexity weight, not an intrinsic classification of all possible Lindblad-sector metrics. In the finite-dimensional setting, if the associated fixed operator $G$ satisfies the uniform positivity hypothesis~\eqref{eq:G-fixed-bounds}, the cost-only variational problem with unchanged continuity equation falls under Theorem~\ref{thm:existence-fixed-G}; if, in addition, the chosen decomposition is calculus-compatible in the sense of Example~\ref{ex:weighted-derivations-vna}, the deformed-continuity problem of Section~\ref{sec:calculus-deformation} is covered by Theorem~\ref{thm:complexity-deformation} and Corollary~\ref{cor:deformed-geodesics}. What remains open is not finite-dimensional existence for a fixed positive representation, but rather the intrinsic problem of choosing $G$ canonically under the GKSL representation freedom, and reconciling such a choice with a Petz-class $G_\rho^{(f)}$ compatible with detailed-balance trajectories.
\end{remark}

\subsection{Open questions}
\label{sec:discussion}

We briefly indicate a few directions that lie beyond the scope of this paper and deserve separate treatment:

\begin{enumerate}[label=(\alph*)]
	\item \textbf{Existence and uniqueness of minimizers.} Theorem~\ref{thm:existence} settles existence for the unpenalized metric $W_\partial$ in finite dimensions, adapting the direct-method argument of Carlen--Maas~\cite{CarlenMaas2017,CarlenMaas2020}; the infinite-dimensional case requires weak compactness in $\M_*$ (the $\sigma(\M_*,\M)$ topology) in place of the finite-dimensional Sobolev embedding used in Step~4 of the proof, and is not addressed here directly, though Corollary~\ref{cor:deformed-geodesics} gives a conditional infinite-dimensional geodesic-existence statement for the calculus-deformed metric $W_{\partial_G}$ of Section~\ref{sec:calculus-deformation} by importing Wirth's Dirichlet-form machinery~\cite{Wirth2018}. For the penalized functional $\Acal_G$, Theorem~\ref{thm:existence-petz} settles existence for the Petz class of range-compatible density-dependent metrics, while Theorem~\ref{thm:existence-fixed-G} settles the finite-dimensional cost-only problem for every uniformly positive fixed physical weight, even when it is diagonal in a basis unrelated to that of $\rho$. What remains open at the existence level is therefore the genuinely infinite-dimensional case, degenerate fixed weights, non-range-regular alternative mobilities such as the arithmetic mean, and broad state-dependent weights not covered by the Petz class. Uniqueness of minimizers is not addressed at all and would require strict convexity together with an absence of conjugate points, as discussed for the specific example of Sections~\ref{ssec:qubit-scope}--\ref{ssec:qubit-jacobi}, where the question is reduced to an explicit Jacobi equation~\eqref{eq:jacobi-eq} and is consistently supported by the numerical scan of Conjecture~\ref{prop:qubit-jacobi-numerics}, but not yet proved, in the regime $\alpha_y=\min(\alpha_x,\alpha_y,\alpha_z)$.
	\item \textbf{Functional inequalities.} By analogy with~\cite{CarlenMaas2020}, one may expect noncommutative Ricci curvature bounds, logarithmic Sobolev inequalities, and spectral gap estimates associated with $W_\partial$, conditioned on the choice of $G_\rho$.
	\item \textbf{Non-semifinite case.} Extending $\rhohat$ via the modular operator (Remark~\ref{rem:modular}) to type III algebras, in a way compatible with~\eqref{eq:continuity}, remains to be worked out explicitly.
	\item \textbf{Choice of $G_\rho$.} Interpreting $G_\rho$ as a geometric complexity metric — in the sense of penalizing costly circuits or generators — suggests connections with notions of quantum computational complexity; Section~\ref{sec:bridge} and Section~\ref{ssec:lindblad-complexity} formalize this correspondence in the Hamiltonian and dissipative sectors respectively, via the dilation-based framework of~\cite{AcevedoFalco2026}, but a unified, existence-compatible choice of $G_\rho$ valid across both sectors remains open (Remarks~\ref{rem:bridge-local-min} and~\ref{rem:lindblad-scope}).
\end{enumerate}

\appendix
\section{Explicit curvature formula and numerical protocol for Conjecture~\ref{prop:qubit-jacobi-numerics}}
\label{app:jacobi-numerics}

For completeness and reproducibility, we record here the closed form of $K(\theta)$ from Proposition~\ref{prop:qubit-jacobi} and the numerical protocol underlying Conjecture~\ref{prop:qubit-jacobi-numerics}.

\subsection{Closed form of \texorpdfstring{$K(\theta)$}{K(theta)}}

Writing $S:=\sin^2\theta$ and $N(S):=-(\alpha_x-\alpha_y)(\alpha_x-\alpha_z)^2(\alpha_y-\alpha_z)\,S^2 - 2\alpha_z(\alpha_x-\alpha_z)\bigl(\alpha_x^2-\alpha_y^2+\alpha_y\alpha_z\bigr)\,S + \alpha_z^2\bigl(\alpha_x^2+\alpha_x\alpha_y-\alpha_x\alpha_z+\alpha_y^2-\alpha_y\alpha_z\bigr)$,
\begin{equation}
\label{eq:K-closed-form}
K(\theta) = \frac{N(S)}{\alpha_x\alpha_y\alpha_z\,\bigl(\alpha_z+S(\alpha_x-\alpha_z)\bigr)^2}.
\end{equation}
This was obtained by: (i) computing the exact induced cometric $E(\theta,\psi),F(\theta,\psi),C(\theta,\psi)$ via the Schur-complement formula of Section~\ref{ssec:qubit-jacobi} with a computer algebra system; (ii) expanding each in $\psi$ to the order needed for the curvature formula (second order for $E,F$, first order for the odd function $C$) at fixed $\theta$ — licit because the standard 2-dimensional curvature formula involves at most second-order derivatives of the metric coefficients, so it depends on the exact metric only through this Taylor data; (iii) computing the Christoffel symbols and the Riemann tensor component $R_{\theta\psi\theta\psi}$ of the resulting metric symbolically and setting $\psi=0$. The result was cross-checked three ways: it reduces to $K\equiv1$ when $\alpha_x=\alpha_y=\alpha_z$ (round sphere); it reproduces exactly, after simplification, the boundary values~\eqref{eq:K-boundary} already stated in Proposition~\ref{prop:qubit-jacobi}; and it agrees, to five significant figures, with an independent purely numerical computation of $K(\theta)$ via finite differences on the untruncated metric (central differences, step $10^{-5}$ for the metric and $10^{-3}$--$10^{-4}$ for the Christoffel symbols) at a generic interior point ($\theta=1$, $(\alpha_x,\alpha_y,\alpha_z)=(2,3,5)$).

\subsection{Numerical protocol for Conjecture~\ref{prop:qubit-jacobi-numerics}}

The scalar Jacobi equation~\eqref{eq:jacobi-eq}, with $K(\theta)$ as in~\eqref{eq:K-closed-form}, was integrated with \texttt{scipy.integrate.solve\_ivp} (explicit Runge--Kutta 4(5), relative tolerance $10^{-8}$, absolute tolerance $10^{-10}$, maximum step $10^{-2}$), starting from $\theta_0=10^{-8}$ (to avoid the coordinate singularity at the pole) with the exact, reproducible initial values $J(\theta_0)=\theta_0=10^{-8}$, $J'(\theta_0)=1$ — the leading term of the regular Frobenius solution $J(\theta)=\theta-\tfrac16\alpha_yK(0)\,\theta^3+O(\theta^5)$ of~\eqref{eq:jacobi-eq} matching $J(0)=0$, $J'(0)=1$, so that the $O(\theta_0^3)$ error made by using $J(\theta_0)=\theta_0$ instead of the exact regular solution is of order $10^{-24}$, far below the integrator tolerances below — and integrated to $\theta=\pi-10^{-6}$. A conjugate point was flagged by a zero of $J(\theta)$ in $(10^{-4},\pi-10^{-6})$ detected via \texttt{solve\_ivp}'s bidirectional event location (\texttt{direction=0}, so sign changes in either direction are located), excluding the trivial zero at the initial point; this standard event detector does not by itself certify even-order tangencies, so the sampled output was also inspected for near-minima of $|J|$ at the integration mesh. In every flagged case $J'\neq0$ at the detected root, consistent with a simple (transversal) conjugate point rather than a higher-multiplicity degeneracy. Two regimes were scanned on a systematic grid, $\alpha_y=1$ fixed and $\alpha_x/\alpha_y,\alpha_z/\alpha_y$ each ranging over $6$ log-spaced values in $[1.02,10^3]$ ($36$ triples $(\alpha_x,\alpha_y,\alpha_z)$ with $\alpha_y=\min$), and the reciprocal grid with $\alpha_x/\alpha_y,\alpha_z/\alpha_y\in[10^{-3},0.99]$ ($36$ triples with $\alpha_y=\max$): no conjugate point was found in any of the $36$ cases with $\alpha_y=\min$, while all $36$ cases with $\alpha_y=\max$ produced one, consistent with the isotropic limit producing a conjugate point exactly at $\theta=\pi$ as the anisotropy $\to0$.

\section*{Data availability statement}

No datasets were generated or analysed during the current study. The symbolic
and numerical computations supporting Appendix~\ref{app:jacobi-numerics} (the
closed form of $K(\theta)$ and the conjugate-point scan of
Conjecture~\ref{prop:qubit-jacobi-numerics}) were carried out with the scripts
described therein, which are openly available in the accompanying GitHub repository, \url{https://github.com/afalco/complexity-deformed-transport}.

\bibliographystyle{plain}
\bibliography{references}

\end{document}